\documentclass[journal]{IEEEtran}

\usepackage[mathscr]{eucal}
\usepackage[cmex10]{amsmath}
\usepackage{epsfig,epsf,psfrag}
\usepackage{amssymb,amsmath,amsfonts,latexsym}
\usepackage{amsmath,graphicx,bm,xcolor,url}
\usepackage[caption=false]{subfig} 
\usepackage{array}
\usepackage{verbatim}
\usepackage{bm}
\usepackage{algorithmic, cite}
\usepackage{algorithm}
\usepackage{verbatim}
\usepackage{textcomp}
\usepackage{mathrsfs}
\usepackage{amsthm}

\catcode`~=11 \def\UrlSpecials{\do\~{\kern -.15em\lower .7ex\hbox{~}\kern .04em}} \catcode`~=13 

\allowdisplaybreaks[3]

\newcommand{\bA}{\mathbf{A}}

\newcommand{\bB}{\mathbf{B}}
\newcommand{\bc}{\mathbf{c}}
\newcommand{\bC}{\mathbf{C}}

\newcommand{\bI}{\mathbf{I}}

\newcommand{\bk}{\mathbf{k}}
\newcommand{\bK}{\mathbf{K}}

\newcommand{\bq}{\mathbf{q}}
\newcommand{\bQ}{\mathbf{Q}}

\newcommand{\bR}{\mathbf{R}}
\newcommand{\bs}{\mathbf{s}}

\newcommand{\bw}{\mathbf{w}}
\newcommand{\bW}{\mathbf{W}}
\newcommand{\bx}{\mathbf{x}}
\newcommand{\bX}{\mathbf{X}}
\newcommand{\by}{\mathbf{y}}
\newcommand{\bY}{\mathbf{Y}}
\newcommand{\bz}{\mathbf{z}}

\DeclareMathAlphabet{\mathbsf}{OT1}{cmss}{bx}{n}
\DeclareMathAlphabet{\mathssf}{OT1}{cmss}{m}{sl}

\DeclareSymbolFont{bsfletters}{OT1}{cmss}{bx}{n}  
\DeclareSymbolFont{ssfletters}{OT1}{cmss}{m}{n}
\DeclareMathSymbol{\bsfGamma}{0}{bsfletters}{'000}
\DeclareMathSymbol{\ssfGamma}{0}{ssfletters}{'000}
\DeclareMathSymbol{\bsfDelta}{0}{bsfletters}{'001}
\DeclareMathSymbol{\ssfDelta}{0}{ssfletters}{'001}
\DeclareMathSymbol{\bsfTheta}{0}{bsfletters}{'002}
\DeclareMathSymbol{\ssfTheta}{0}{ssfletters}{'002}
\DeclareMathSymbol{\bsfLambda}{0}{bsfletters}{'003}
\DeclareMathSymbol{\ssfLambda}{0}{ssfletters}{'003}
\DeclareMathSymbol{\bsfXi}{0}{bsfletters}{'004}
\DeclareMathSymbol{\ssfXi}{0}{ssfletters}{'004}
\DeclareMathSymbol{\bsfPi}{0}{bsfletters}{'005}
\DeclareMathSymbol{\ssfPi}{0}{ssfletters}{'005}
\DeclareMathSymbol{\bsfSigma}{0}{bsfletters}{'006}
\DeclareMathSymbol{\ssfSigma}{0}{ssfletters}{'006}
\DeclareMathSymbol{\bsfUpsilon}{0}{bsfletters}{'007}
\DeclareMathSymbol{\ssfUpsilon}{0}{ssfletters}{'007}
\DeclareMathSymbol{\bsfPhi}{0}{bsfletters}{'010}
\DeclareMathSymbol{\ssfPhi}{0}{ssfletters}{'010}
\DeclareMathSymbol{\bsfPsi}{0}{bsfletters}{'011}
\DeclareMathSymbol{\ssfPsi}{0}{ssfletters}{'011}
\DeclareMathSymbol{\bsfOmega}{0}{bsfletters}{'012}
\DeclareMathSymbol{\ssfOmega}{0}{ssfletters}{'012}

\DeclareMathOperator*{\argmax}{arg\,max}

\DeclareMathOperator{\tr}{tr}

\DeclareMathOperator{\supp}{supp}

\DeclareMathOperator{\cov}{\mathsf{Cov}}

\newtheorem{theorem}{Theorem} 
\newtheorem{lemma}[theorem]{Lemma}

\newtheorem{data model}{Data Model}

\newcommand{\qednew}{\nobreak \ifvmode \relax \else
      \ifdim\lastskip<1.5em \hskip-\lastskip
      \hskip1.5em plus0em minus0.5em \fi \nobreak
      \vrule height0.75em width0.5em depth0.25em\fi}

\DeclareMathOperator{\lcm}{lcm}
\DeclareMathOperator{\gap}{gap}
\DeclareMathOperator{\SNR}{SNR}

\IEEEpeerreviewmaketitle

\begin{document} 
\bstctlcite{IEEEexample:BSTcontrol}

\title{Detection of Brain Stimuli Using Ramanujan Periodicity Transforms}
\author{Pouria~Saidi,~Azadeh~Vosoughi and George~Atia\\%
Department of Electrical and Computer Engineering, University of Central Florida, Orlando, Fl, 32816 USA\\
pouria.saidi@knights.ucf.edu; azadeh@ucf.edu; george.atia@ucf.edu}

\maketitle

\begin{abstract}  
\textit{Objective:}
The ability to efficiently match the frequency of the brain's response to repetitive visual stimuli in real time is the basis for reliable SSVEP-based Brain-Computer-Interfacing (BCI).
\textit{Approach:} The detection of different stimuli is posed as a composite hypothesis test, where SSVEPs are assumed to admit a sparse representation in a Ramanujan Periodicity Transform (RPT) dictionary. 
For the binary case, we develop and analyze the performance of an RPT detector based on a derived generalized likelihood ratio test.
Our approach is extended to multi-hypothesis multi-electrode settings, where we capture the spatial correlation between the electrodes using pre-stimulus data. We also introduce a new metric for evaluating SSVEP detection schemes based on their achievable efficiency and discrimination rate tradeoff for given system resources.
\textit{Results:} We obtain exact distributions of the test statistic in terms of confluent hypergeometric functions.
Results based on extensive simulations with both synthesized and real data indicate that the RPT detector substantially outperforms spectral-based methods. Its performance also surpasses the state-of-the-art Canonical Correlation Analysis (CCA) methods with respect to accuracy and sample complexity in short data lengths regimes crucial for real-time applications. The proposed approach is asymptotically optimal as it closes the gap to a perfect measurement bound as the data length increases.
In contrast to existing supervised methods which are highly data-dependent, the RPT detector only uses pre-stimulus data to estimate the per-subject spatial correlation, thereby dispensing with considerable overhead associated with data collection for a large number of subjects and stimuli.
\textit{Significance:} Our work
advances the theory and practice of emerging real-time BCI and affords a new framework for comparing SSVEP detection schemes across a wider spectrum of operating regimes. 
\end{abstract}


\section{Introduction} \label{sec:Intro} \IEEEPARstart{B}{rain} computer interfacing (BCI) is used to connect the human nervous system to external devices. Pushing the frontiers of such technology holds promise to assist, and improve the quality of life of, patients with motor disability due to various neurological disorders through the use of intention-controlled prosthetics. 
Non-invasive BCIs exploit features from Electroencephalogram (EEG) signals distinctive of different tasks. For example, motor imagery BCI refers to a class of BCI in which features and patterns are extracted from EEG sensorimotor rhythms triggered by action imagination (e.g., imagined movement of the limbs). Despite noteworthy efforts to devise a host of techniques to classify EEG trials associated with movement imagination~\cite{shin2012sparse,saidi2015Globalsip}, 
motor imagery BCI remains at its infancy to date. 


Evoked-potential-based BCI is another non-invasive technology that records EEG signals from electrical activity induced on the cortex in response to some repetitive visual stimulus using electrodes attached to the scalp. The brain's responses to such external stimuli, known as Steady State Visual Evoked Potential (SSVEPs), are known to exhibit periodicity matching that of the stimuli~\cite{liavas1998periodogram}.   
The (quasi) periodic patterns of SSVEPs have been the basis for much progress in the theory and practice of SSVEP-BCI. For example, spectral-based methods capturing the energy distribution across the frequency spectrum such as Power Spectral Density Analysis (PSDA) are popular choices for classifying SSVEPs \cite{PSDA2005,liavas1998periodogram}. However, the presence of high levels of background noise due to brain chatter -- inevitably superimposed on the recorded brain's response -- continues to be a major challenge in face of the development of  reliable SSVEP detectors. Its degrading effect on performance is further compounded by the stringent and indispensable latency requirements of real-time BCI, in which high levels of accuracy are mandated with short delay. For instance, the frequency resolution of spectral-based methods is severely diminished with the use of shorter data lengths.  

An alternative and state-of-the-art approach to SSVEP detection relies on Canonical Correlation Analysis (CCA) -- initially proposed in \cite{CCA_original} to find relations between sets of variates. In particular, the authors in \cite{lin2006frequency} leveraged CCA in developing a method for SSVEP classification, which we refer to in this work as standard CCA. The key idea underlying standard CCA is obtaining the maximum correlation of the data with a reference matrix defined for each class consisting of periodic signals with the frequency of the stimulus and its harmonics to decide on a class label. 
Some of the state-of-the-art supervised algorithms originated from standard CCA such as in \cite{multiwaycca,multiwayl1,itcca,FCCA} to further boost the performance of SSVEP classification, albeit at the expense of heavy reliance on post-stimulus training trials  
\cite{nakanishi2015comparison}. One drawback of excessive reliance on such data in the training phase of supervised methods, such as Individual Template CCA (IT CCA), lies in the overhead and cost associated with data collection. For example, the low-frequency flashing lights in SSVEP-based BCIs used to obtain large amplitude responses on the brain cortex \cite{lowfreqstimuli} could become tiring and burdensome for some subjects, may cause eye fatigue, and may even trigger seizures in some patients \cite{photicsensitivity,highfreqstimuli}. 

Other than the use of Fourier transform and periodograms to identify periodic signals intrinsic to time series-data (as in PSDA), there exist periodicity estimation methods that search for periodicities and regularities in data directly in the time domain~ \cite{muresan2003orthogonal,periodicitytransform1999}. 
For example, the authors in \cite{tenneti2015nested} studied periodicity estimation using representations of discrete periodic sequences in Nested Periodic Matrices (NPMs). They also introduced the so-called Ramanujan Periodicity Transforms (RPT) as an instance of NPMs.  The use of RPT was shown to exhibit robustness to noise and phase shifts. Ramanujan sums defining the bases for RPT were shown useful in representing periodic sequences in various applications~\cite{planat2009}. However, only few works have investigated their application with biomedical signals, such as in the analysis of T-wave alternans \cite{mainardi2008analysis} and the detection of tandem repeats in DNA \cite{tenneti2016detecting}. We leveraged an RPT-based model to detect SSVEPs for the first time in \cite{saidi2017detection,Pouria_asilomar2017} providing preliminary results for the current work, and demonstrated its ability to capture the underlying periodicity in SSVEPs and its robustness to latencies naturally present in the brain's response to external stimuli. 
%

\noindent\textbf{Contributions:} In this paper, we build on our preliminary prior work on the RPT model and extensively analyze SSVEP detection in a composite hypothesis testing framework. The following summarizes the main contributions of this paper. 
\begin{itemize}
\item We develop an RPT detector of the brain's response associated with various stimuli based on a generalized likelihood ratio test (GLRT) under the proposed RPT model. 
\item We provide an exact analysis of the performance of the RPT detector by deriving the distributions of the test statistic characterized in terms of confluent hypergeometric functions for the binary case.
\item We devise flexible Gaussian approximations of the derived distributions that avail an efficient framework for the design of stimulus waveforms for high-accuracy BCIs.
\item We establish asymptotic optimality of the proposed approach by analyzing the performance gap with respect to a derived perfect measurement bound for the composite test. The gap is shown to vanish exponentially fast in the EEG system resources (data length $L$ and signal to noise ratio ($\SNR$)).
\item We introduce and investigate the tradeoff between the error exponent, $-\log P_e/ (L\cdot\SNR)$, capturing the rate at which the classification error decays with $L$ and $\SNR$ and the discrimination rate $\log_2 M$ (the logarithm of the number of classes being discriminated) for various detection methods including the RPT detector for a given number of electrodes. Inspired by the concept of diversity and multiplexing from communication theory \cite{zheng2003diversity}, this tradeoff is introduced here for the first time in the context of SSVEP detection and is shown to be particularly useful for comparing various detection methodologies across an entire spectrum of operation regimes.
\item We extend our approach and derive the corresponding test for multi-hypothesis and multi-channel settings, where we capture the spatial correlation between the recording electrodes. 
\end{itemize}

The paper is organized as follows. In section \ref{sec:Met}, we provide a brief background about Ramanujan sums and RPT matrices and their properties. Then, we present the composite hypothesis testing model, and provide an analysis of the RPT detector and the associated sufficient statistic for the binary case. The analysis is extended to multi-class and multi-electrode settings. We present our results on synthesized and real data in Section \ref{sec:Res}. Section \ref{sec:disc} is devoted to a discussion and our concluding remarks are in Section \ref{sec:conc}.
\section{Methods} \label{sec:Met}
\textit{Notation:} We use lowercase letters for scalars, bold lowercase letters for vectors and bold uppercase letters for matrices. We use $d|T$ to indicate that $d$ is a divisor of $T$. The Euler totient function of $p$, that is the number of positive integers smaller than $p$ that are co-prime to $p$, is denoted by $\phi(p)$.
Given a set $S$, the set $S^c$ denotes its complement. The operator $\tr(.)$ denotes the trace of its matrix argument, and $\log$ denotes the natural logarithm to the base $e$, unless the base is made explicit. We use the notation $\bx \sim \mathcal{N}(\mu, \Sigma)$ to indicate that a random vector $\bx$ has a multivariate normal (Gaussian) distribution with mean vector $\mu$ and covariance matrix $\Sigma$. 

\subsection{Ramanujan Sums and RPT Dictionary}
\label{sec:Ram}

In \cite{tenneti2015nested}, the authors introduced NPMs that can capture periodicity in sequences. They extended the notion of these matrices to periodicity dictionaries. In \cite{tenneti2016unified}, periodicity dictionaries of order $P_{\text{max}}$ are defined as the set of signals $\mathcal{B}$  that can represent all periodic sequences with period $1 \leq p \leq P_{\text{max}}$ through linear combinations of signals in the set $\mathcal{B}$. It was shown that a dictionary that spans all subspaces of periodic signals of periods $1$ to $P_{\text{max}}$ must contain at least $\phi(p)$ linearly independent signals with period $p$ for each $p$. Consequently, a periodic dictionary of order $P_{\text{max}}$ must have at least $\sum_{p=1}^{P_{\text{max}}}\phi(p)$ linearly independent signals.

RPT matrices are instances of the NPMs built from Ramanujan sums \cite{ramanujan1918sum}. A Ramanujan sum is defined as 
\begin{equation}
\begin{aligned}
c_q(n) = \sum_{\substack{k=1\\(k,q)=1}}^qe^{j2\pi kn/q}\:,
\end{aligned}
\end{equation}
where $(k,q)$ is the greatest common divisor (gcd) of $k$ and $q$. The sequence $c_q(n)$ is an all integer, symmetric and periodic sequence with period $q$. For instance, $c_1(n) = \{1\}$, $c_2(n) = \{1,-1\}$, $c_3(n) = \{2,-1,-1\}$, $c_4(n) = \{2,0,-2,0\}$ and $c_5(n) = \{4,-1,-1,-1,-1\}$. These examples show only one period of the sequences. Properties of Ramanujan sums are investigated in \cite{vaidyanathan2014ramanujan1} and \cite{vaidyanathan2014ramanujan2}. One important feature of $c_q(n)$ highly relevant to our work is the orthogonality property.  
Specifically, the sequences $c_{q_1}(n)$ and $c_{q_2}(n)$ are orthogonal over the sequence length $L = \lcm(q_1,q_2)$ for $q_1\ne q_2$, where $\lcm$ is the least common multiplier of $q_1$ and $q_2$, i.e.,
\begin{equation}
\begin{aligned}
\sum_{n=0}^{L-1} c_{q_1}(n) c_{q_2}(n-k) = 0, \quad q_1 \neq q_2
\end{aligned}
\label{eq:orth_prop}
\end{equation}
for any integer $0\leq k\leq L-1$.

We define the sequence $\bc_q$ 
\[
\bc_q = 
\begin{bmatrix}
c_q(0) & c_q(1) & \ldots & c_q(q-1)
\end{bmatrix}^T.
\]
For each $1\leq q \leq P_{\text{max}}$, a submatrix $\bC_q$ is constructed from columns that are circularly downshifted versions of $\bc_q$, i.e.,
\begin{equation}
\begin{aligned}
\bC_q=&
\begin{bmatrix}
\bc_q & \bc_q^{(1)} & \dots & \bc_q^{(\phi(q)-1)}
\end{bmatrix},
\end{aligned}
\end{equation}
where $\bc_q^{(1)}$ is the circularly downshifted version of $\bc_q$, i.e., 
\begin{equation}
\begin{aligned}
\bc_q^{(1)}=&
\begin{bmatrix}
c_q(q-1) & c_q(0) & c_q(1)&\dots & c_q(q-2)
\end{bmatrix}^T.
\end{aligned}
\end{equation}
The versions $\bc_q^{(i)}$ are similarly defined with higher order downshifts. For instance, $\bC_4$ and $\bC_5$ are
\begin{equation}
\begin{aligned}
\bC_4=&
\begin{bmatrix}
2 & 0\\0 & 2\\-2 & 0\\ 0 & -2
\end{bmatrix}_{4\times \phi(4)}
\bC_5=&
\begin{bmatrix}
4&-1&-1&-1\\-1&4&-1&-1\\-1&-1&4&-1\\-1&-1&-1&4\\-1&-1&-1&-1
\end{bmatrix}_{5\times \phi(5)}
\end{aligned}
\end{equation}
We extend the matrices $\bC_q$ periodically to length $L$ yielding submatrices $\bR_q$. We  can readily define the RPT dictionary matrix $\bK$ \cite{tenneti2016unified} constructed by concatenating all submatrices $\bR_q, q  = 1, \ldots, P_{\text{max}}$, where $P_{\text{max}}$ is the largest possible value that $q$ assumes. Thus,
\begin{equation}
\begin{aligned}
\bK=&
\begin{bmatrix}
\bR_1 & \bR_2 &\bR_3& \dots & \bR_{P_{\text{max}}}
\end{bmatrix}.
\end{aligned}
\end{equation}
Such dictionaries are called Nested Periodic Dictionaries (NPD). We note that NPDs have exactly $\phi(p)$ signals with period $p$ for each $p$. In the next two sections, we develop the SSVEP detection problem in a composite hypothesis testing framework using the signal representations in an RPT dictionary and analyze the performance of the developed detector.

\subsection{Binary SSVEP Detection} \label{sec:Binary_case}
\subsubsection{Composite binary hypothesis testing model} \label{sec:Binary_model}
In this section, we assume that the SSVEP is associated with one of two stimuli with frequencies $f_0$ and $f_1$, which represent two possible hypotheses $H_0$ and $H_1$, respectively. Equivalently, the recorded sequence is of period $T_0$ or $T_1$. We consider an observation model in which we express the measurements under each of the hypotheses using the model,
\begin{equation}
\begin{aligned}
H_0:& ~\by = \bK \bx_0 + \bw\\
H_1:& ~\by = \bK \bx_1 + \bw\\
\end{aligned}
\label{eq:model_binary_TBME}
\end{equation}
where $\by$ is an $L \times 1$ observation (measurement) vector, $\bK$ is the $L \times \sum_{p=1}^{P_{max}}\phi(p)$ RPT dictionary matrix ($L$ is the length of the sequences), $\bx_0$ and $\bx_1$ are the $\sum_{p=1}^{P_{max}}\phi(p) \times 1$ representations of the observations in the RPT subspace, modeled as two deterministic but unknown vectors. The $L\times 1$ vector $\bw$ is an additive white Gaussian noise (AWGN) vector with covariance matrix $\sigma^2\bI$ modeling the background EEG noise, that is,
$\bw \sim \mathcal{N}\left(0,\sigma^2\bI\right)$. Given $\bx_m$ we can write the distribution of $\by$ as $\by \sim \mathcal{N}\left(\bK \bx_m,\sigma^2\bI\right)$, where $m \in\{0,1\}$.
Reference \cite{tenneti2015nested} considers a period \emph{estimation} problem where the support set (locations of the nonzero entires) of the signal representation is unknown. By contrast, here we incorporate prior information about the support of $\bx_0$ and $\bx_1$ under both hypotheses since the period of the SSVEP should match 
that of the stimulus. 
We denote the known support sets $\supp(\bx_m)$ for $\bx_m$ by $S_m :=  \{j:x_{m,j}\ne 0 \}$, where $x_{m,j}$ is the $j$-th element of $\bx_m$, $m = 0,1$. The non-zero elements correspond to the columns of the submatrices of $\bK$ associated with period $T_m$ and its divisors. We define $\bK_{S_m}$ whose columns are  the columns of $\bK$ indexed by the support set $S_m$, and column vector $\bx_{S_m}$ whose entries are equal to the non-zero entries of $\bx_m$ indexed by $S_m$. We define the  $\SNR$ corresponding to hypothesis $H_m$ as
\begin{equation}
\begin{aligned}
\SNR_m = 
\frac{\bx_{S_m}^T\bK_{S_m}^T\bK_{S_m}\bx_{S_m}}{\sigma^2L}.
\end{aligned}
\label{eq:SNR}
\end{equation}
Without loss of generality, we set $\sigma^2 = 1$ so that different SNRs can be accounted for by varying the magnitude of the signal part on the support.

\subsubsection{Binary RPT detector} \label{sec:Binary_detector}
Let $f(\by|H_m)$ denote the conditional probability density function (pdf) of the measurement vector $\by$ given hypothesis $H_m$. For the model in  (\ref{eq:model_binary_TBME}) we can write: 
\begin{equation}
\begin{aligned}
&f\left(\by|H_m\right) =\frac{1}{(2\pi\sigma^2)^{\frac{L}{2}}}\exp\left(\frac{-\|\by-\bK\bx_{m}\|_2^2}{2\sigma^2}\right).
\end{aligned}
\label{eq:dist_binary_TBME}
\end{equation}
The binary RPT detector is obtained from the GLRT defined as \cite{van2004detection}
\begin{equation}
\begin{aligned}
\mathcal{L}(\by)=&\frac{\max\limits_{\bx_1:\supp(\bx_1) = S_1} f(\by|H_1)}{\max\limits_{\bx_0:\supp(\bx_0) = S_0} f(\by|H_0)} \: \underset{H_0}{\overset{H_1}{\gtrless}} \: \eta \:,
\end{aligned}
\label{eq:GLRT_binary_TBME}
\end{equation}
which requires computing the restricted Maximum Likelihood (ML) estimate of $\bx_m$ when the support set is restricted to $S_m$. 
From (\ref{eq:dist_binary_TBME}), the restricted ML estimate under $H_m, m = 0,1$, is obtained as the solution to
\begin{equation}
\begin{aligned}
\min_{\bx_m:\supp(\bx_m) = S_m} \|\by - \bK \bx_m\|_2^2.
\end{aligned}
\label{eq:restriced_MLE_binary_TBME}
\end{equation}
Accordingly, we can rewrite (\ref{eq:restriced_MLE_binary_TBME}) as
\begin{equation}
\begin{aligned}
\min \|\by - \bK_{S_m} \bx_{S_m}\|_2^2.
\end{aligned}
\label{eq:restriced_MLE_binary_TBME_rewrite}
\end{equation}
The solution to (\ref{eq:restriced_MLE_binary_TBME_rewrite}) can be expressed as
\begin{equation}
\begin{aligned}
\hat{\bx}_{S_m} &= (\bK_{S_m}^T \bK_{S_m})^{-1}\bK_{S_m}^T \by\:. \
\end{aligned}
\label{Psuedoinverse_binary_TBME}
\end{equation}
Replacing in (\ref{eq:GLRT_binary_TBME}), the decision rule reduces to
\begin{align}
\label{eq:Decision_rule_general_binary_TBME}
\ell(\by)=\by^T \bB \by - \by^T \bA \by 
\underset{H_0}{\overset{H_1}{\gtrless}} \: 2\sigma^2\log \eta = \gamma\:,
\end{align}
where $\bB=\bK_{S_1}(\bK_{S_1}^T\bK_{S_1})^{-1}\bK_{S_1}^T, \bA=\bK_{S_0}(\bK_{S_0}^T\bK_{S_0})^{-1}\bK_{S_0}^T$ and $\ell(\by)$ is the sufficient statistic. The two matrices $\bA$ and $\bB$ in (\ref{eq:Decision_rule_general_binary_TBME}) are idempotent, i.e., $\bA^2 = \bA, \bB^2 = \bB$ and their eigenvalues are either 0 or 1.

\subsubsection{Performance analysis of the binary RPT detector} \label{sec:performance}
In this section, we analyze the performance of the detector in (\ref{eq:Decision_rule_general_binary_TBME}). We start off with the special case where the length of $\by$ (measured in number of sample) is $L = \lcm(T_0,T_1)$. In this case, the orthogonality of the RPT sub-matrices is preserved. Under this assumption, we are able to obtain the exact distributions of the test statistic $\ell(\by)$, and in turn provide an exact performance analysis. Then, we consider the general case in which $\by$ is of arbitrary length $L$. In this case, the orthogonality of the RPT submatrices is not necessarily preserved.  For the general case, we provide an approximate analysis based on Gaussian approximations of the distributions of $\ell(\by)$. 

\noindent$\bullet$\textit{Orthogonal case when $L = \lcm(T_0,T_1)$:} 
\label{subsec:special_case}
If the periods $T_0$ and $T_1$ share any divisors, the support sets $S_0$ and $S_1$ are non-disjoint, i.e., $S_0 \cap S_1 \ne \emptyset$. Therefore, the decision rule in (\ref{eq:Decision_rule_general_binary_TBME}) reduces to 
\begin{equation}
\begin{aligned}
\ell(\by) = \by^T \bB^{\perp} \by - \by^T \bA^\perp \by\: \underset{H_0}{\overset{H_1}{\gtrless}} \: 2\sigma^2\log \eta = \gamma\: ,
\end{aligned}
\label{eq:dec_orthogonal_TBME}
\end{equation}
where the matrices $\bA^{\perp}$ and $\bB^{\perp}$ are given by
\begin{equation}\label{eq:orhtogonal_A_B_TBME}
\begin{aligned}
\bA^{\perp} =& ~\bK_{S_0\setminus S_1} ( \bK_{S_0\setminus S_1}^T \bK_{S_0\setminus S_1}  )^{-1}  \bK_{S_0\setminus S_1}^T\\
\bB^{\perp} =& ~\bK_{S_1\setminus S_0} ( \bK_{S_1\setminus S_0}^T \bK_{S_1\setminus S_0}  )^{-1}  \bK_{S_1\setminus S_0}^T \:, 
\end{aligned}
\end{equation}
and $S_i \setminus S_j$ denotes the difference set of $S_i$ and $S_j$. The superscript $^\perp$ is reserved for this orthogonal case throughout this subsection. Fig. \ref{fig:matrices}(a) shows a schematic for an example of the dictionary matrix $\bK$ and its submatrices in which $T_0=10$ and $T_1=8$. Fig. \ref{fig:matrices}(b) illustrates the RPT matrices that are restricted to support sets $S_0$ and $S_1$ and also submatrices $\bK_{S_0 \setminus S_1}$ and $\bK_{S_1 \setminus S_0}$ corresponding to two difference sets $S_0 \setminus S_1$ and $S_1\setminus S_0$, respectively. As shown, the submatrices $\bR_1$ and $\bR_2$ corresponding to the divisors 1 and 2, respectively, exist in both $\bK_{S_0}$ and $\bK_{S_1}$, thus do not play a role in the sufficient statistic. Per the definitions of $\bA^{\perp}$ and $\bB^{\perp}$ in (\ref{eq:orhtogonal_A_B_TBME}), the measurement vector $\by$ is projected onto the column space of $\bK_{S_0\setminus S_1}$ and $\bK_{S_1\setminus S_0}$, and the decision rule in (\ref{eq:dec_orthogonal_TBME}) chooses the index of the subspace that yields the larger projection.
This insight is demonstrated in Fig. \ref{fig:orthogonal_subspace}. In Fig. \ref{fig:orthogonal_subspace}(a), we have $L = \lcm (T_0,T_1)$ and thus the orthogonality of the RPT subspaces is preserved, and in Fig. \ref{fig:orthogonal_subspace}(b) $L$ is arbitrary and therefore the subspaces are not orthogonal. We have the following lemma. 
%

 \begin{figure}
 \includegraphics[width=\linewidth,height=5cm]{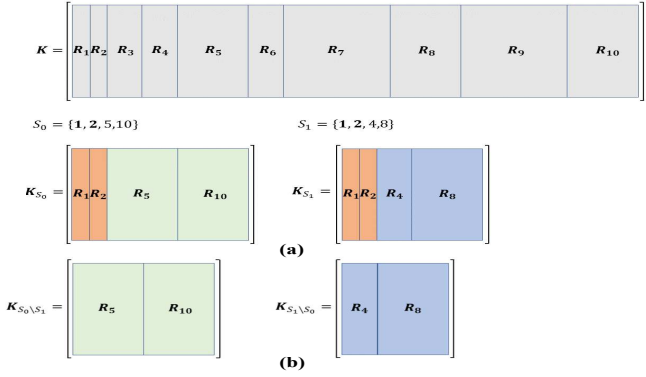}
\vspace{-.5cm}
\caption{(a) Schematic for an example of the RPT dictionary matrix $\bK$ and its submatrices, where $T_0=10$ and $T_1=8$.  (b) RPT matrices are restricted to support sets $S_0$ and $S_1$. In this example, $\bR_1$ and $\bR_2$ exist in both $\bK_{S_0}$ and $\bK_{S_1}$ matrices.}
\label{fig:matrices}
\vspace{-0.5cm}
 \end{figure}

\begin{figure}
\centering
\includegraphics[width=\linewidth,height=2.3cm]{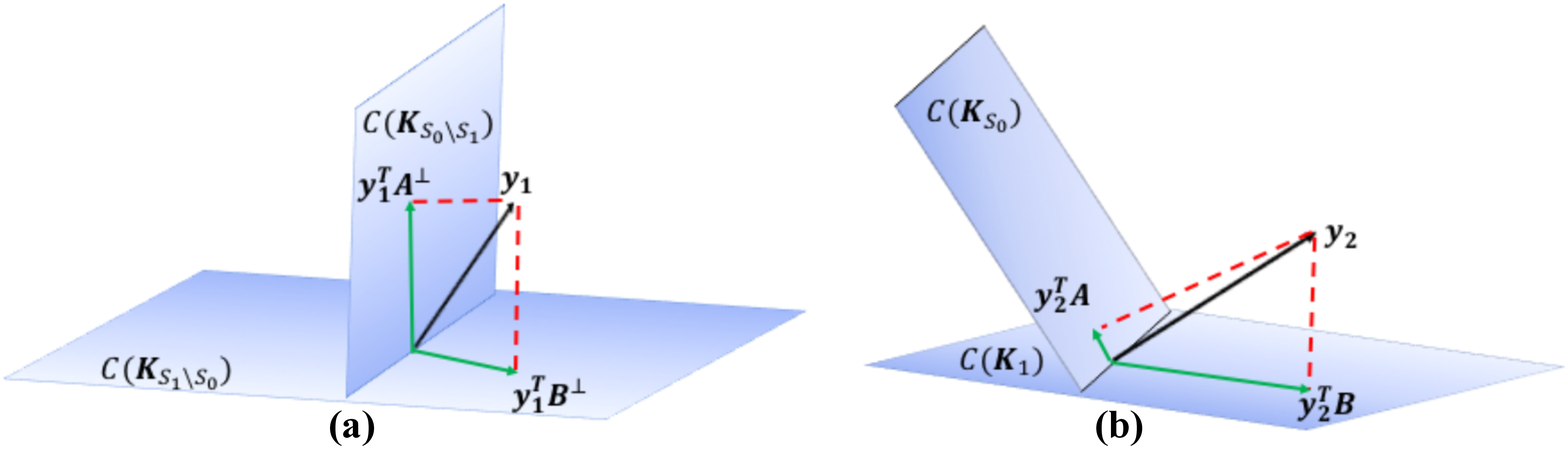}
\vspace{-.4cm}
\caption{(a) For $L = \lcm(T_0,T_1)$, the column spaces ${\cal C}(\bK_{S_0\setminus S_1})$ and ${\cal C}(\bK_{S_1\setminus S_0})$ of the orthogonal submatrices effective in the decision rule are orthogonal. The GLRT chooses the index corresponding to the subspace that yields the larger projection. Here, $\by_1^T\bA^{\perp} \by_1 >\by_1^T\bB
^{\perp} \by_1$, hence the RPT detector maps $\by_1$ to $H_0$. (b) For arbitrary $L$, the submatrices are not orthogonal. Since, $\by_2^T\bA\by_2 < \by_2^T\bB
 \by_2$, the RPT detector maps $\by_2$ to $H_1$.}
\label{fig:orthogonal_subspace}
\vspace{-0.5cm}
 \end{figure}
\begin{lemma} \label{lem:A_B_independent}
\textit{For $L = \lcm(T_0,T_1)$, we have $\bA^{\perp}\bB^{\perp} = \mathbf{0}$.}
\end{lemma}

\begin{proof} Lemma \ref{lem:A_B_independent} follows from the orthogonality property in (\ref{eq:orth_prop}), which for the choice of $L = \lcm(T_0,T_1)$ ensures that the columns of submatrices $\bK_{S_0}$ and $\bK_{S_1}$ corresponding to distinct divisors are orthogonal.
Since the submatrices $\bK_{S_0 \setminus S_1}$ and $\bK_{S_1 \setminus S_0}$ of the RPT dictionary $\bK$ indexed by the difference sets do not share any columns, we have $\bK_{S_0\setminus S_1}^T ~\bK_{S_1\setminus S_0} =0$, hence, $\bA^{\perp}\bB^{\perp} = 0$.
\end{proof}
To analyze the performance of the RPT detector, we need $f(\ell(\by)|H_0)$ and $f(\ell(\by)|H_1)$, the pdfs of the test statistic $\ell(\by)$ under each hypothesis. From \cite[Lemma 1.1]{anderson1980cochran} it follows that the two quadratic terms $\by^T \bA^{\perp} \by$ and  $\by^T \bB^\perp \by$ in (\ref{eq:dec_orthogonal_TBME}) have Chi-squared distributions $\chi^2(r_A^{\perp},\lambda_{m,A}^{2,\perp})$ and $\chi^2(r_B^{\perp},\lambda_{m,B}^{2,\perp})$ under hypothesis $H_m$, where $r_A^{\perp}$ and $r_B^{\perp}$ are the degrees of freedom, and $\lambda_{m,A}^{2,\perp}$ and $\lambda_{m,B}^{2,\perp}$ are the non-centrality parameters. From \cite{press1966linear}, we can readily state the following theorem proved in Appendix \ref{App:proof_thr_dist}, which provides exact expressions for the distributions of the sufficient statistic $\ell(\by)$ for the orthogonal case when $L = \lcm(T_0,T_1)$. 
\textit{
\begin{theorem}\label{thr:dist_theorem}(Distribution of test statistic for orthogonal case)
The distributions $f\left(\ell(\by)|H_m\right), m = 0,1$, of the sufficient statistic $\ell(\by) = \by^T \bB^{\perp} \by - \by^T \bA^{\perp} \by$ in (\ref{eq:dec_orthogonal_TBME}) are given by 
\begin{equation}\label{eq:dist_theory_TBME}
\begin{aligned}
f\left(\ell(\by)|H_1\right) = \sum_{i=0}^{\infty} \frac{\exp(-\frac{1}{2}\lambda_{1,B}^{2,\perp})(\lambda_{1,B}^{2,\perp}/2)^i}{i!} p_{r_B^{\perp}+2i,r_A^{\perp}}(t)\\
f\left(\ell(\by)|H_0\right) = \sum_{j=0}^{\infty} \frac{\exp(-\frac{1}{2}\lambda_{0,A}^{2,\perp})(\lambda_{0,A}^{2,\perp}/2)^j}{j!} p_{r_B^{\perp},r_A^{\perp}+2j}(t)\\
\end{aligned}
\end{equation}
where
%
%
\begin{equation}\label{eq:dist_final_TBME}
\begin{aligned}
p_{a,b}(t)
=
\begin{cases}
\frac{2^{\frac{-(a+b)}{2}}}{\Gamma(a/2)}t^{\frac{a+b-2}{2}}e^{\frac{-t}{2}} \psi(\frac{b}{2},\frac{a+b}{2};t) \quad  \text{if} \hspace{1mm} t\geq 0 \\
\\
\frac{2^{\frac{-(a+b)}{2}}}{\Gamma(b/2)}(-t)^{\frac{a+b-2}{2}}e^{\frac{t}{2}} \psi(\frac{a}{2},\frac{a+b}{2};-t) \quad  \text{if} \hspace{1mm} t\leq 0
\end{cases}
\end{aligned}
\end{equation}
and $\psi(a,b;x)$ is defined as 
\begin{equation}\label{eq:hyper_TBME}
\begin{aligned}
\psi(a,b;x) = (\Gamma(a))^{-1} \int_0 ^{\infty} e^{-xt} t^{a-1} (1+t)^{b-a-1} dt\:,
\end{aligned}
\end{equation}
%
%
%
and the parameters of the non-central Chi-squared distributions are
\begin{equation} \label{eq:char_ortg_TBME}
\begin{aligned}
\lambda_{0,A}^{2,\perp}=& \sum_{\substack{i=1\\d_i \notin \mathcal{T}_1}}^v  \bx_{S_0,i}^{T}\bK_{S_0,i}^T\bK_{S_0,i}\bx_{S_0,i}~,\quad \lambda_{1,A}^{2,\perp}=0\\
\lambda_{1,B}^{2,\perp} =&
\sum_{\substack{j=1,\\d_j \notin \mathcal{T}_0}}^u \bx_{S_1,j}^T\bK_{S_1,j}^{T}\bK_{S_1,j}\bx_{S_1,j}~,\quad \lambda_{0,B}^{2,\perp}=0\\
r_A^{\perp}=&\sum_{\substack{i=1\\d_i \notin \mathcal{T}_1}}^v\phi(d_i|T_0) \qquad 
r_B^{\perp}=\sum_{\substack{j=1\\d_j \notin \mathcal{T}_0}}^u\phi(d_j|T_1)\\
\end{aligned}
\end{equation}
where $\mathcal{T}_0$ and $\mathcal{T}_1$ are the sets of divisors of $T_0$ and $T_1$ of cardinalities $v$ and $u$, respectively, $\bK_{S_0,i}, i =1, \ldots, v$, and $\bK_{S_1,j}, j =1, \ldots, u$, are the corresponding submatrices, and $\bx_{S_m,i}$ and $\bx_{S_m,j}$ are rows of $\bx_{S_m}$ restricted to the index set of the $i$-th and $j$-th divisors of $T_0$ and $T_1$. 
\end{theorem}
}

%
\begin{proof}
Per Lemma \ref{lem:A_B_independent} and Theorem \ref{thr:dist_theorem}, the sufficient statistic $\ell(\by)$ in (\ref{eq:dec_orthogonal_TBME}) is the difference of two independent non-central Chi-squared random variables (RVs). Accordingly, the proof of Theorem \ref{thr:dist_theorem} follows from Theorems 2.1B, 3.1, 3.2, 3.3 in \cite{press1966linear} which characterize the distribution of linear combinations of independent non-central Chi-squared RVs in terms of confluent hypergeometric functions. Further details are provided in Appendix \ref{App:proof_thr_dist}.
\end{proof}
In designing BCIs, it is desirable to select waveforms that yield well-separable hypotheses with respect to the employed signals representation. Due to the complexity of the exact pdf expressions of the sufficient statistic $\ell(\by)$ in (\ref{eq:dist_theory_TBME}), we approximate these pdfs with Gaussian distributions by fitting the first and second order statistics. 
These Gaussian approximations lead to a flexible waveform design. Recalling that a non-central Chi-squared distribution with $r$ degrees of freedom and non-centrality parameter $\lambda^2$ has a mean $r + \lambda^2$ and variance $2(r+2\lambda^2)$, the approximate Gaussian pdfs under $H_0$ and $H_1$ are 
\begin{equation}\label{eq:dist_approx_ortg_TBME}
\begin{aligned}
H_0: \ell(\by) \sim & \mathcal{N}\left(r_B^{\perp} -r_A^{\perp}-\lambda_{0,A}^{2,\perp}, 2(r_B^{\perp} + r_A^{\perp}) + 4\lambda_{0,A}^{2,\perp}\right)\\
H_1: \ell(\by) \sim & \mathcal{N}\left(r_B^{\perp}-r_A^{\perp}+\lambda_{1,B}^{2,\perp},2(r_B^{\perp} + r_A^{\perp}) + 4\lambda_{1,B}^{2,\perp}\right)
\end{aligned}
\end{equation}
To validate these theoretical results, we consider (\ref{eq:model_binary_TBME}) and assume that the periods corresponding to the two stimuli are $T_0 = 32$, $T_1 = 18$ under hypotheses $H_0$ and $H_1$, respectively. We fix $\SNR = -14$ dB and generate 2000 observation vectors $\by$ under each hypothesis. 
Fig.~\ref{fig:Gaussian_apprx}(a) shows that the exact pdfs of $\ell(\by)$ under $H_0$ and $H_1$ derived in (\ref{eq:dist_theory_TBME}) agree with the approximated Gaussian pdfs in (\ref{eq:dist_approx_ortg_TBME}) and histograms of the sufficient statistic $\ell(\by)$ obtained from numerical experiments.
\begin{figure}
\includegraphics[width=\linewidth,height=4cm]{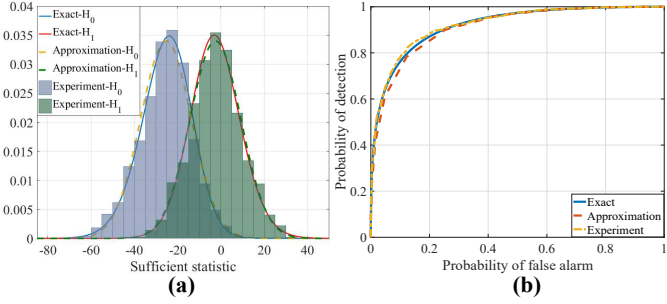}
\vspace{-.65cm}
\caption{(a) The exact and Gaussian approximated distributions of $\ell(\by)$ derived in (\ref{eq:dist_theory_TBME}) and (\ref{eq:dist_approx_ortg_TBME}) are in agreement and match the histogram of $\ell(\by)$ obtained from numerical experiments. (b) The ROC curve obtained from numerical integration of the exact pdfs is close to the ones obtained from the Gaussian approximated pdfs and from the numerical experiment.}
\label{fig:Gaussian_apprx}
\vspace{-0.5cm}
 \end{figure}
 Based on these Gaussian approximations, we can obtain the probability of detection $P_D : = \mathbb{P}_1(\ell(\by) > \gamma)$  and the probability of false alarm $P_F := \mathbb{P}_0({\ell(\by) > \gamma})$, where $\mathbb{P}_m$ denotes the probability measure under the $m$-th hypothesis
\begin{equation}\label{eq:P_D_bayes_special_TBME}
\begin{aligned}
P_D &=Q\left(\frac{\gamma -r_B^{\perp}+r_A^{\perp} -\lambda_{1,B}^{2,\perp}}{\sqrt{(2(r_B^{\perp}+r_A^{\perp})+4\lambda_{1,B}^{2,\perp}}}\right)\\
P_F &=Q\left(\frac{\gamma -r_B^{\perp}+r_A^{\perp}+\lambda_{0,A}^{2,\perp}}{\sqrt{(2(r_B^{\perp}+r_A^{\perp})+4\lambda_{0,A}^{2,\perp}}}\right)\\
\end{aligned}
\end{equation}
where $Q(.)$ denotes the Q-function, which is the tail of the standard Normal distribution. Fig. \ref{fig:Gaussian_apprx}(b) illustrates the ROC curves, where $P_D$ and $P_F$ corresponding to the exact pdfs are obtained using numerical integration of the exact pdfs in (\ref{eq:dist_theory_TBME}), and $P_D$ and $P_F$ corresponding to the approximate Gaussian pdfs in (\ref{eq:dist_approx_ortg_TBME}) are found using (\ref{eq:P_D_bayes_special_TBME}). Clearly, the ROC curves are in close agreement. 

\noindent$\bullet$\textit{General case when $L$ is arbitrary:} \label{subsec:arbitrary_L}
%
In this section, we treat the general case with arbitrary $L$ that does not necessarily lead to orthogonality. The two quadratic terms $\by^T \bA \by$ and $\by^T \bB \by$ have $\chi^2$ distributions $\chi^2(r_A,\lambda_{m,A}^2)$ and $\chi^2(r_B,\lambda_{m,B}^2)$ under hypotheses $H_m, m = 0, 1$. The parameters of these distributions can be written as
\begin{equation}
\begin{aligned}
\lambda_{0,A}^2=& \bx_{S_0}^{T}\bK_{S_0}^T\bA \bK_{S_0}\bx_{S_0}= \bx_{S_0}^{T}\bK_{S_0}^T \bK_{S_0}\bx_{S_0}\\
\lambda_{0,B}^2=&\bx_{S_0}^{T}\bK_{S_0}^T\bB \bK_{S_0}\bx_{S_0}\\
\lambda_{1,A}^2=& \bx_{S_1}^{T}\bK_{S_1}^T\bA \bK_{S_1}\bx_{S_1}\\
\lambda_{1,B}^2=&\bx_{S_1}^{T}\bK_{S_1}^T\bB \bK_{S_1}\bx_{S_1} = \bx_{S_1}^{T}\bK_{S_1}^T \bK_{S_1}\bx_{S_1}\\
r_{A} =& \sum_{i=1}^v \phi(d_i|T_0)= T_0 \qquad r_{B} = \sum_{j=1}^u \phi(d_j|T_1)=T_1
\end{aligned}
\label{eq:chi_char_TBME}
\end{equation}
Since $L$ is arbitrary, the two random variables $\by^T \bA \by$ and $\by^T \bB \by$ are not necessarily orthogonal, wherefore they are not independent and finding the exact pdf of $\ell(\by)$ requires obtaining the joint pdf of these two dependent/correlated RVs.  
Akin to the orthogonal case, we 
approximate the pdf of $\ell(\by)$ with a Gaussian distribution by fitting the first and second order statistics and considering the covariance between the two RVs $\by^T \bA \by$ and $\by^T \bB \by$
for the test statistic under each hypothesis. Therefore, we consider the approximate model
\begin{equation}\label{eq:dist_suff_stat_binary_TBME}
\begin{aligned}
H_0: \ell(\by) \:\sim \: & \mathcal{N}\left(r_{B}-r_{A}+\lambda_{0,B}^2 -\lambda_{0,A}^2,\sigma_{0}^2\right)  \\
H_1: \ell(\by) \:\sim \: & \mathcal{N}\left(r_{B}-r_{A}+\lambda_{1,B}^2 -\lambda_{1,A}^2,\sigma_{1}^2\right) \\
\end{aligned}
\end{equation}
where 
\begin{equation}\label{eq:var_TBME}
\begin{aligned}
\sigma_0^2 =  2(r_{B} + 2\lambda_{0,B}^2) & +2(r_{A} + 2\lambda_{0,A}^2) \\
& - 2\cov(\by^T \bB \by,\by^T \bA \by|H_0)\\
\sigma_{1}^2 = 2(r_{B} + 2\lambda_{1,B}^2) & +2(r_{A} + 2\lambda_{1,A}^2) \\ 
& - 2\cov(\by^T \bB \by,\by^T \bA \by|H_1)
\end{aligned}
\end{equation}
and $\cov(\by^T \bB \by,\by^T \bA \by|H_m)$ denotes the covariance between $\by^T \bA\by$ and $\by^T \bB \by$ under $H_m$. In Appendix \ref{App:proof_cov}, we show that 
\begin{equation}\label{eq:cross_cov_TBME}
\begin{aligned}
\cov(\by^T \bB \by,\by^T \bA \by|H_0) = 4 \lambda_{0,B}^2 +2 \sum_{i=1}^L\sum_{j=1}^L c_{ij}\\
\cov(\by^T \bB \by,\by^T \bA \by|H_1) = 4 \lambda_{1,A}^2 + 2\sum_{i=1}^L\sum_{j=1}^L c_{ij}
\end{aligned}
\end{equation}
where $c_{ij}$ is the $(i,j)$-th entry of matrix $\bC= \bA \odot \bB$, and $\odot$ is an element-wise product. We can readily obtain general expressions for $P_D$ and $P_F$,
\begin{equation}\label{eq:P_D_bayes_binary_TBME}
\begin{aligned}
&P_D = \\ 
&Q\left(\frac{\gamma -r_B -\lambda_{1,B}^2+r_A +\lambda_{1,A}^2}{\sqrt{(2(r_B+r_A)+4(\lambda_{1,B}^2-\lambda_{1,A}^2)-4\sum_{i=1}^L\sum_{j=1}^L c_{ij}}}\right)\\
&P_F = \\
&Q\left(\frac{\gamma -r_B-\lambda_{0,B}^2+r_A+\lambda_{0,A}^2}{\sqrt{(2(r_B+r_A)+4(\lambda_{0,A}^2-\lambda_{0,B})-4\sum_{i=1}^L\sum_{j=1}^L c_{ij}}}\right).\\
\end{aligned}
\end{equation}
Under the Neyman-Pearson criteria \cite{van2004detection}, the maximum value for $P_D$ for false alarm level $\alpha$, i.e., $P_F \leq \alpha$, is given in (\ref{eq:P_D_NP_binary_TBME}).
\begin{table*}[!t]
\begin{align}\label{eq:P_D_NP_binary_TBME}
P_D =Q\left(\frac{Q^{-1}(\alpha)\sqrt{2(r_B+r_A)+4(\lambda_{0,A}^2-\lambda_{0,B}^2)-4\sum_{i=1}^L\sum_{j=1}^L c_{ij}}+\lambda_{0,B}^2+\lambda_{1,A}^2-\lambda_{0,A}^2-\lambda_{1,B}^2}{\sqrt{(2(r_B+r_A)+4(\lambda_{1,B}^2-\lambda_{1,A}^2)-4\sum_{i=1}^L\sum_{j=1}^L c_{ij}}}\right)
\end{align}
\end{table*}
\begin{table*}[!t]
\begin{align} \label{eq:P_D_PMB_TMBE}
P_{D_{\text{PMB}}} = Q\left(Q^{-1}(\alpha)-\sqrt{\bx_{S_1}^T\bK_{S_1}^T\bK_{S_1}\bx_{S_1}+\bx_{S_0}^T\bK_{S_0}^T\bK_{S_0}\bx_{S_0}-\bx_{S_1}^T\bK_{S_1}^T\bK_{S_0}\bx_{S_0}-\bx_{S_0}^T\bK_{S_0}^T\bK_{S_1}\bx_{S_1}}\right)
\end{align}
\hrulefill
\end{table*}
\subsubsection{Perfect measurement bound for the binary RPT detector}
In this section, we derive an upper bound on the performance of the proposed binary RPT detector. We use a standard approach from detection theory in which an upper bound on the performance of all composite tests is obtained by assuming that the signals $\bx_0$ and $\bx_1$ are known, hence the appellation `perfect measurement bound' (PMB) \cite{van2004detection}. Under this hypothetical assumption, the problem reduces to one of a simple binary hypothesis testing problem. 
For our problem, under this assumption the optimal test becomes
\begin{equation} \label{eq:dec_rule_PMB_TMBE}
\begin{aligned}
\ell(\by) & = \by^T (\bK_{S_1}\bx_{S_1} - \bK_{S_0}\bx_{S_0})\\
&\underset{H_0}{\overset{H_1}{\gtrless}}\ln(\eta) + \frac{1}{2}(\bx_{S_1}^T\bK_{S_1}^T\bK_{S_1}\bx_{S_1} - \bx_{S_0}^T\bK_{S_0}^T\bK_{S_0}\bx_{S_0}).
\end{aligned}
\end{equation}
Given the model in (\ref{eq:model_binary_TBME}), under the Neyman-Pearson criteria, $P_D$ for false alarm level $\alpha$
is expressed in (\ref{eq:P_D_PMB_TMBE}). We note that, (\ref{eq:P_D_PMB_TMBE}) provides an upper bound on the performance of all composite tests corresponding to the model in (\ref{eq:model_binary_TBME}), including the proposed RPT detector in (\ref{eq:Decision_rule_general_binary_TBME}).

To compare the performance of the proposed RPT detector to the PMB, we let $\gap(L,\SNR)$ denote the difference between $P_D$ of the RPT detector and the PMB. 
Lemma \ref{lem:Gap} provides an approximate characterization of the gap in the asymptotic regime of large $L$. There are two main sources of approximation, namely using the Gaussian approximate distribution and using an approximation for the Q-function.  
\textit{
\begin{lemma} \label{lem:Gap}
Suppose the $\SNR$s corresponding to hypotheses $H_0$ and $H_1$, defined in (\ref{eq:model_binary_TBME}), are equal. For large $L$ and $\SNR$, the difference between $P_D$ of the RPT detector in (\ref{eq:P_D_NP_binary_TBME}) and the hypothetical detector (which corresponds to the PMB) in (\ref{eq:P_D_PMB_TMBE}) is well-approximated by
\begin{equation}\label{Gap_TBME}
\begin{aligned}
\gap(L,\SNR)
\hspace{-2pt}=\hspace{-2pt} P_D \hspace{-2pt}-\hspace{-2pt} P_{D_{\text{PMB}}} \hspace{-2pt}\approx\hspace{-2pt}\frac{e^{-\frac{L.\SNR}{2}}\left(\sqrt{2}-e^{-\frac{L.\SNR}{2}}\right)}{2\sqrt{\pi}\sqrt{L.\SNR}}\:.
\end{aligned}
\end{equation}
Accordingly,
\begin{equation}\label{Gap_order_TBME}
\begin{aligned}
\lim_{L\cdot\SNR\rightarrow\infty}\frac{|\log\gap|}{L\cdot \SNR}  = O(1)\:, 
\end{aligned}
\end{equation}
where $O(1)$ is a constant\footnote{This is the standard Big O notation, so that $f(x) = O(g(x))$ iff there exists positive numbers $C$ and $x_0$ such that $|f(x)| \leq C g(x), \forall x\geq x_0$.}.
\end{lemma}
}
\begin{proof}
See Appendix \ref{App:lemma_gap}
\end{proof}
Lemma \ref{lem:Gap} indicates that, assymptotically the gap is dominated by the exponentially decaying function, establishing the asymptotic optimality of the proposed RPT detector, in the sense that the performance gap with respect to the PMB -- which provides an upper bound on the performance of any composite test -- approaches zero. 

For a numerical example, we consider (\ref{eq:model_binary_TBME}) and assume that the periods corresponding to the two stimuli are $T_0 = 32$ and $T_1 = 18$ under hypotheses $H_0$ and $H_1$, respectively. We fix $\SNR = -15$dB and generate 5000 observation vectors $\by$ under each hypothesis. Fig. \ref{fig:GAP_TBME}(a) depicts the difference between the PMB and $P_D$ for the RPT detector in (\ref{eq:Decision_rule_general_binary_TBME}) as we vary $L$ from $\max\{T_0, T_1\}$ to 2000. The figure confirms that the numerical experiments are in agreement with the analytical Gaussian approximations and verifies the asymptotic optimality of the RPT detector. The constant slope of Fig. \ref{fig:GAP_TBME}(b) corroborates the linear scaling of $|\log\gap|$ with respect to the product $L\cdot\SNR$ as we derived in (\ref{Gap_order_TBME}).
\begin{figure}
\centering
\includegraphics[width=\linewidth,height=4cm]{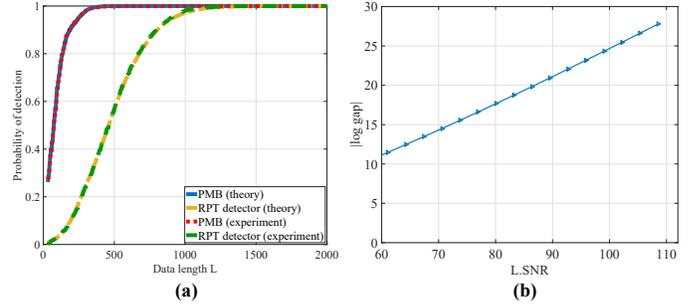}
\vspace{-.65cm}
\caption{(a) PMB and $P_D$ of the binary RPT detector from both theory and numerical experiment. (b) $|\log \gap|$ scales linearly in $L.\SNR$.}
\label{fig:GAP_TBME}
\vspace{-0.5cm}
\end{figure} 

\subsection{Multi-class and Multi-electrode SSVEP Detection} \label{sec:M-ary case}
\subsubsection{Composite M-ary hypothesis testing model} \label{sec:M_ary model}
Our model in (\ref{eq:model_binary_TBME}) can be extended to distinguishing $M>2$ hypotheses (multiple classes)
$H_m$ corresponding to $M$ stimuli and periodic brain responses with frequencies $f_m, m\in \mathcal{M}= \{0,1,\cdots,M-1\}$. 
We can also take into account the measurements collected (recorded) from multiple electrodes. Since these electrodes lie in close proximity, their signals are not independent. Hence, it is important that such an extended model also captures the spatial correlation between the measurements of different electrodes.
We extend the model in (\ref{eq:model_binary_TBME}) to $M$ hypotheses, in which the measurements under each hypothesis are modeled as 
\begin{equation}
\begin{aligned}
H_0:& ~\bY = \bK \bX_0 + \bW\\
H_1:& ~\bY = \bK \bX_1 + \bW\\
\vdots\\
H_{M-1}:& ~ \bY = 
\bK\bX_{M-1}+\bW
\end{aligned}
\label{eq:M_class_model_TBME}
\end{equation}
where $\bY$ is an $L\times N_c$ observation (measurement) matrix, $L$ is the length of the recorded data, $N_c$ is the number of electrodes (channels), and $\bX_m, m\in\mathcal{M}$ are the signal representations of the observations in the RPT subspace. Similar to the composite binary hypothesis testing model in Section \ref{sec:Binary_case}, the signal representations $\bX_m, m\in\mathcal{M}$, are modeled as $M$ deterministic but unknown matrices
and we restrict the row support of the signal representation $\bX_m$ under the $m$-th hypothesis to the atoms of the RPT dictionary matrix $\bK$ that span the subspace of periodic signals with period $T_m$. 
We denote the known support sets $\supp(\bX_m)$ for $\bX_m$ by $S_m := \{j: \bx_{m,j} \ne \mathbf{0}\}$, where $\bx_{m,j}$ is the $j$-th row of $\bX_m, m = 0, \ldots, M-1$. Let $\by_l, \bk_l$ and $\bw_l$ denote the $l$-th row of $\bY, \bK$ and $\bW$, respectively. From the extended model in (\ref{eq:M_class_model_TBME}) we have
\begin{equation}
\begin{aligned}
\by_l = \bk_l\bX_m + \bw_l, ~l = 1, \ldots, L, ~~ m\in \mathcal{M}\:.
\end{aligned}
\end{equation}
The noise vector $\bw_l$ is assumed to be Gaussian with covariance matrix $\mathbf{\Sigma_w}$, i.e., $\bw_l\sim \mathcal{N}(\mathbf{0},\mathbf{\Sigma_{w}})$, where $\mathbf{\Sigma_w}$ captures the spatial correlation between the measurements of different electrodes. Under these assumptions, we have 
$\by_l \sim \mathcal{N}(\bk_l\bX_m,\mathbf{\Sigma_{w}})$. Also, we assume the vectors $\bw_l, l = 1,\ldots, L$, are independent and identically distributed.
\subsubsection{Multi-class RPT detector} 
\label{sec:rpt_detec}
Let $f\left(\bY|H_m\right)$ denote the conditional pdf of the measurement matrix $\bY$ given hypothesis $H_m$. For the extended model in (\ref{eq:M_class_model_TBME}) we can write:
\begin{equation}
\begin{aligned}
&f\left(\bY|H_m\right) = \\
\hspace{-3pt} &\prod_{l=1}^{L}\hspace{-4pt} \frac{1}{(2\pi)^{\frac{N_c}{2}}|\mathbf{\Sigma_{w}}|^\frac{1}{2}}\hspace{-2pt}\exp\hspace{-3pt}\left(\hspace{-4pt}\frac{-\hspace{-2pt}\left(\by_l-\bk_l\bX_{m}\right)\hspace{-2pt}{\mathbf{\Sigma_w}}^{-1}\hspace{-2pt}\left(\by_l-\bk_l\bX_{m}\right)^T}{2}\hspace{-2pt}\right)\hspace{-3pt}.
\end{aligned}
\label{eq:dist_multi_elec_TBME}
\end{equation}
Similar to the binary RPT detector in Section \ref{sec:Binary_case}, the multi-class RPT detector is obtained 
from the GLRT assuming uniform prior probabilities $P(H_m)=\frac{1}{M}$ and uniform cost assignment \cite{van2004detection}.
We define the generalized likelihood ratios (GLRs)
\begin{equation}
\begin{aligned}
\mathcal{L}_m(\bY) \triangleq & \frac{\max\limits_{\bX_m:\supp(\bX_m) = S_m } f\left(\bY|H_m\right)}{\max\limits_{\bX_0:\supp(\bX_0) = S_0} f\left(\bY|H_0\right)},\quad m=1,\dots,M-1,
\end{aligned}
\label{eq:M_ary_GLRs_TBME}
\end{equation}
which require computing the restricted ML estimate of $\bX_m$ when the support set is restricted to $S_m$ for $m \in \mathcal{M}$. Let $\hat\bX_m$ denote the ML estimate of $\bX_m$ when the row support set is restricted to $S_m$. 
From (\ref{eq:dist_multi_elec_TBME}), we find that the ML estimate $\hat\bX_m$ is the solution to the following
\begin{equation}
\begin{aligned}
\min_{\bX_m:\supp(\bX_m) = S_m} \tr\left(\left(\bY-
\bK \bX_{m}\right)\mathbf{\Sigma_w}^{-1}\left(\bY-\bK \bX_{m}\right)^T\right) 
\end{aligned}
\label{eq:ML_X}
\end{equation}
such that $\bx_{m,j} = \mathbf{0}, \forall j\in S_{m}^c$.
Accordingly, (\ref{eq:ML_X}) can be rewritten as
\begin{equation}
\begin{aligned}
\min \tr\left(\left(\bY-
\bK_{S_m} \bX_{S_m}\right)\mathbf{\Sigma_w}^{-1}\left(\bY-\bK_{S_m} \bX_{S_m}\right)^T\right)
\end{aligned}
\label{eq:ML_X_S_TBME}
\end{equation}
where $\bK_{S_m}$ and $\bX_{S_m}$ are the columns of $\bK$ and rows of $\bX_m$ indexed by $S_m$, respectively.  
The solution to (\ref{eq:ML_X_S_TBME}) can be written as (for proof see Appendix \ref{App:est})
\begin{equation}
\begin{aligned}
\hat{\bX}_{S_m} = \left(\bK_{S_m}^{T}\bK_{S_m}\right)^{-1}\bK_{S_m}^{T}\bY.
\end{aligned}
\label{eq:estimator_multi}
\end{equation}

Thus, the decision $\hat m$ is obtained by replacing (\ref{eq:dist_multi_elec_TBME}) in (\ref{eq:M_ary_GLRs_TBME}), which gives
\begin{align}
&\hat m\left(\bY\right) = \argmax\limits_{m\in\mathcal{M}}\hspace{-3pt}\nonumber\\
&
\hspace{-2mm}\tr\left(\bY\mathbf{\Sigma_w}^{-1}\hat{\bX}_{S_m}^{T}\bK_{S_m}^{T}-\frac{\bK_{S_m}\hat{\bX}_{S_m}\mathbf{\Sigma_w}^{-1}\hat{\bX}_{S_m}^{T}\bK_{S_m}^{T}}{2} \right).
\label{eq:multi_class_decision_TBME}
\end{align}
Replacing with the estimates in (\ref{eq:estimator_multi}), we reach
\begin{align}
\hat m\left(\bY\right) = \argmax_{m\in\mathcal{M}} \tr\left(\bY\mathbf{\Sigma_w}^{-1}\bY^T\bA_m\right),
\label{eq:multi_elec_m_class_TBME}
\end{align}
where $\bA_m$ is
\begin{equation}
\begin{aligned}
\bA_m =\bK_{S_m}\left(\bK_{S_m}^{T}\bK_{S_m}\right)^{-1}\bK_{S_m}^{T}.
\end{aligned}
\end{equation}

\subsection{Error Exponent - Discrimination Rate Tradeoff}
It is conceivable that a reliable detection scheme should be able to leverage on information obtained from additional electrodes to boost the system's efficiency, given that the detector is tasked with distinguishing between a fixed number of classes. Alternatively, for the same efficiency, the additional electrodes may be exploited to 
increase the detector's ability to discriminate a larger number of classes. This insight is the basis for a new tradeoff, introduced here for the first time in the context of SSVEP detection, between `error exponent' and the `discrimination rate' defined in our problem setup as $-\log P_e/ (L\cdot\SNR)$ and $\log_2 M$, respectively. Inspired by the notion of diversity and multiplexing tradeoff in information theory \cite{zheng2003diversity}, our goal is to provide a suitable metric associated with a broader spectrum of operating regimes across which different SSVEP detection methods could be compared. Specifically, given the system resources $L$ and $\SNR$, a detection method can leverage a multi-electrode EEG system (with $N_c$ electrodes) to trade the system's efficiency (measured by the error probability $P_e$ associated with the detector) for a higher discrimination rate (the detector's ability to discriminate $M$ classes, measured by $\log_2M$) and vice-versa.  
The error probability $P_e$ in the binary hypothesis testing setting with uniform prior probabilities $P(H_0)=P(H_1)=1/2$ is $P_e = \frac{1}{2}\left(1-P_D+P_F\right)$, thus it can be approximated in the high SNR regime as
\begin{equation}\label{P_e_TBME}
\begin{aligned}
P_e
=\frac{e^{-\frac{L.\SNR}{8}}}{\sqrt{2\pi}\sqrt{L.\SNR/4}}\:,
\end{aligned}
\end{equation}
using the expressions of $P_D$ and $P_F$ derived in (\ref{eq:P_D_bayes_binary_TBME}) and the Q-function approximation. Since $-\log P_e/ (L\cdot\SNR)$ captures the rate at which the error probability decays with $L$ and $\SNR$, we use the logarithm of the error probability as a measure of the system's efficiency. For distinguishing between $M$ classes, we define the average error probability as 
\begin{equation}
\begin{aligned}
P_e = \frac{1}{M}\sum_{m=1}^M P(H_j|H_m) \quad \text{for } j\neq m. 
\end{aligned}
\end{equation}
 
As a result, we can characterize the error probability for an entire spectrum of regimes corresponding to different discrimination rates for a given number of electrodes $N_c$ and system resources (in terms of data length $L$ and $\SNR$). 
 
\section{Numerical Results} \label{sec:Res}
In this section, we start off by presenting some numerical results from experiments with synthesized data to verify and validate the theory developed for the binary RPT detector of Section \ref{sec:Binary_case}. Then, we provide results for the multi-class and multi-electrode RPT detector using both synthesized and real data, along with comparisons to the existing SSVEP detection methods, including PSDA, standard CCA and IT CCA \cite{itcca}. As real data, we employ a publicly available SSVEP dataset \cite{nakanishi2015comparison}, which contains multiple goal frequencies (frequencies of stimuli) $f_m, m=0, ..., M-1$, with a sampling frequency of $f_s=256$ Hz, and $M$ is the number of classes to be discriminated. The data length $L$ (measured in samples) is related to the data length $T$ (measured in seconds) through $T=L/f_s$.
We study the classification accuracy $A = 1-P_e$ and the Information Transfer Rate (ITR) defined in \cite{nakanishi2015comparison} in terms of the number of classes to be discriminated $M$, the classification accuracy $A$, and the data length $T$ in seconds, 
\begin{equation}
\begin{aligned}
\text{ITR} \hspace{-2pt}=\hspace{-3pt}\left(\log_2 M\hspace{-1pt} +\hspace{-1pt} A\log_2A + (1-A)\log_2\hspace{-3pt}\Big(\frac{1-A}{M-1}\Big)\hspace{-3pt}\right)\hspace{-3pt}\Big(\frac{60}{T}\Big).\nonumber
\end{aligned}
\end{equation}
We further discuss the advantage of NPMs as finite complete bases for estimating the periods of periodic sequences and study the newly introduced error exponent-discrimination rate tradeoff in terms of the number of electrodes $N_c$.

\subsection{Theoretical Validation of Binary RPT Detector Using Synthesized Data}

Fig.~\ref{fig:Gaussian_apprx} validated our theoretical results for the binary RPT detector for the orthogonal case when $L=\lcm(T_0,T_1)$. Here, we validate  our theoretical results for the general case when $L$ is arbitrary. We consider (\ref{eq:model_binary_TBME}) and fix $\SNR = - 15$ dB and generate 2500 observation vectors $\by$  under each hypothesis. We consider two cases: for case (a) 
we assume that the periods corresponding to the two stimuli are $T_0 = 15$  and $T_1 = 25$ under hypotheses $H_0$  and $H_1$, respectively, and for case (b) we assume $T_0 = 32$  and $T_1 = 18$.
Fig.~\ref{fig:ROC_TBME} illustrates the ROC curves for these two cases, where $P_D$ and $P_F$ are obtained using both the expressions derived in (\ref{eq:P_D_bayes_binary_TBME}) and the numerical experiment.  Fig.~\ref{fig:P_e_TBME}(a) and Fig.~\ref{fig:P_e_TBME}(b) show the error probability $P_e$, and the classification accuracy $A$ in terms of the data length $L$ for the setup in case (b). These two figures show that results from both theory and experiments strongly agree.  
\begin{figure}
\centering
\includegraphics[width=\linewidth,height=4cm]{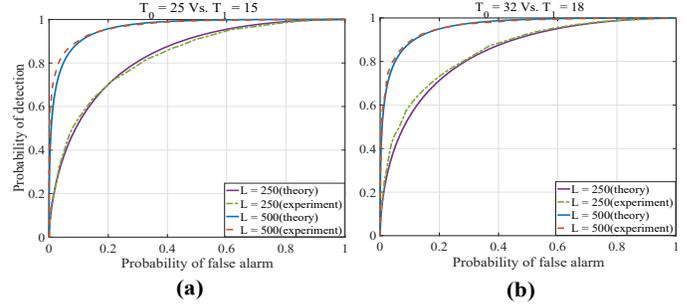}
\vspace{-.65cm}
\caption{ROC curves show $P_D$ versus $P_F$ for both theory (Gaussian approximated pdfs) and numerical experiment. (a)  $T_0 = 25$, $T_1 = 15$, (b) $T_0 = 32$, $T_1 = 18$.}
\label{fig:ROC_TBME}
\end{figure}
\begin{figure}
\centering
\includegraphics[width=\linewidth,height=4cm]{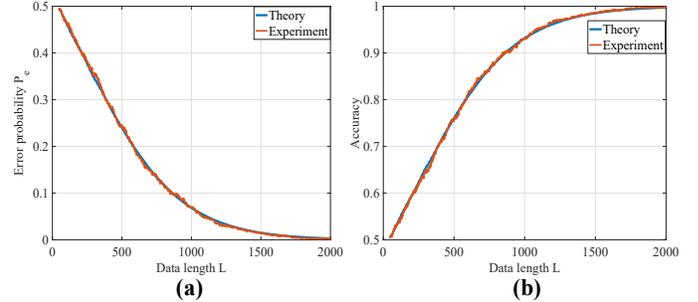}
\vspace{-.65cm}
\caption{$T_0 = 32$ and $T_1 = 18$ (a) error Probability $P_e$ is obtained using both theory (Gaussian approximated pdfs) and numerical experiment. (b) Classification accuracy versus data length $L$ (in number of samples).}
\label{fig:P_e_TBME}
\vspace{-0.5cm}
\end{figure}

\subsection{Real Data Description for Validation of Multi-class and Multi-electrode RPT Detector} The public SSVEP daatset \cite{nakanishi2015comparison} we use is generated based on the sampling frequency of $f_s=256$ Hz.
Ten subjects have participated in the experiment and 15 trials have been recorded for each subject per goal frequency. Hence, the dataset has a total number of 150 trials for each goal frequency. We use 9 goal frequencies (i.e., the number of classes to be discriminated is $M=9$) for performance validation of the RPT detector and performance comparison among different SSVEP detection methods. These 9 goal frequencies are $f_0=9.25$, $f_1=9.75$, $f_2=10.25$, $f_3=10.75$, $f_4=11.25$, $f_5=11.75$, $f_6=12.25$, $f_7=12.75$ and $f_8=14.25$ Hz.
All the recorded data is filtered using 4-30 Hz bandpass filters. The trials are recorded using 8 electrodes positioned on the occipital region of the subject's brain cortex as depicted in Fig.~\ref{fig:location}. 
The period $T_m$ corresponding to each goal frequency $f_m$ is computed as $T_m=f_s/f_m$, then is rounded to its nearest integer value. 
\begin{figure}
\centering
 \includegraphics[width=4.5cm,height=4.5cm]{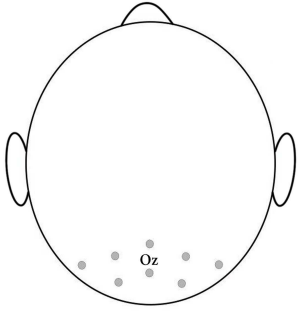}
 \vspace{-0.5cm}
\caption{Eight electrodes are positioned on the occipital region of the subject's brain cortex.}
\label{fig:location}
 \end{figure}
 
\subsection{Validation of Multi-class and Multi-electrode RPT Detector Using Synthesized and Real Data}
In this section, we investigate the performance of the multi-class and multi-electrode RPT detector obtained in (\ref{eq:multi_elec_m_class_TBME}).
Since all the electrodes are located on the occipital region of the brain cortex, the collected data is normally correlated. As described in Section \ref{sec:M-ary case}, our model captures the spatial correlation between the data of different electrodes through the covariance matrix $\mathbf{\Sigma_w}$. 
We use \emph{pre-stimulus} data to find a sample covariance matrix and use it as an estimate of the true $\mathbf{\Sigma_w}$. 
An advantage of using only pre-stimulus data segments in estimating $\mathbf{\Sigma_w}$ is that one does not need to collect data from each subject under each goal frequency $f_m$. This is important in practice as it significantly reduces the time, cost, and overhead associated with data collection, especially with a large number of classes and subjects.

\noindent$\bullet$\textit{Latency and its effect on the performances of SSVEP detection methods:}
\begin{figure}
\vspace{0.5cm}
\centering
\includegraphics[width=8.5cm,height=3.5cm]{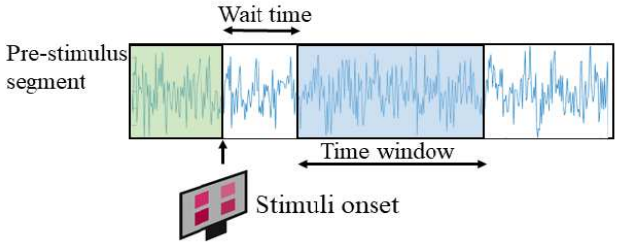}
\vspace{-0.3cm}
\caption{Schematic for an SSVEP observation. There is an unknown latency between the stimulus onset and the brain's response.}
\label{fig:latency_proc}
\vspace{-0.7cm}
\end{figure}
\begin{figure}
\vspace{0.5cm}
\centering
\includegraphics[width=\linewidth,height=4cm]{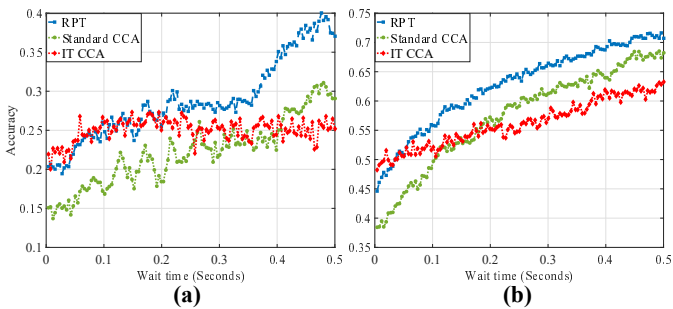}
\caption{Classification accuracy with (a) 0.5 second post-stimulus SSVEPs and (b) 1 second post-stimulus SSVEPs. The wait time represents the time between the onset of the stimulus and the beginning of our time window intended to capture the time it takes for the brain to respond to the external stimulus.}
\label{fig:latency}
\end{figure}
It is well recognized that there exists a delay (latency) between the onset of an external stimulus and the beginning of the brain's response\cite{FCCA}. This delay widely varies among different subjects and is generally unknown. The latency is an important factor that affects the performance of all SSVEP detection methods. To test the performance of the RPT detector with respect to this uncontrolled latency, we define Wait time and Time window (See Fig. \ref{fig:latency_proc}). Wait time is the time between the stimulus onset and the beginning of Time window. Time window is the portion of the recorded data we are using for our performance evaluation with length $T$ measured in seconds. 
Fig.~\ref{fig:latency}(a) and Fig.~\ref{fig:latency}(b) compare the performance of the RPT detector to those of standard CCA and IT CCA using 0.5 second and 1 second of post-stimulus data as the wait time varies, respectively. Fig.~\ref{fig:latency} demonstrates that the performance of all methods improve as we increase the wait time due to an improved $\SNR$.
Noting that the computational complexity and performance of all these methods depend on the data length underscores an important advantage of the RPT approach. In particular, taking the latency into account one could attain a better performance for a short data length which is crucial for real-time BCI. As shown, for the short data lengths used, the RPT method achieves higher accuracy. Based on these results we consider a 0.25 second wait time after the onset of the stimulus.
\smallbreak
\noindent$\bullet$\textit{Effect of spatial correlation knowledge on the RPT detector performance:} Fig.~\ref{fig:M_class_N_channel_TBME}(a) presents the classification accuracy $A$ as function of the data length $T$ for a multi-class multi-electrode RPT detector using synthesized data.
In this experiment, we generate periodic sequences with $N_c = 8$ electrodes for $M = 9$ classes. The dashed blue curve corresponds to a genie-aided RPT detector that knows $\mathbf{\Sigma_w}$. The dotted red curve is obtained when the RPT detector uses an estimate of $\mathbf{\Sigma_w}$ obtained from pre-stimulus data (real recorded data from one of the subjects). The dashed green curve is for a RPT detector that falsely assumes no spatial correlation between the recorded data from different electrodes and classifies the data thereof, hence this case is denoted as `model mismatch'.
\noindent$\bullet$\textit{Performance comparison of different SSVEP methods using real data:} Figs.~\ref{fig:M_class_N_channel_TBME}(b) and \ref{fig:M_class_N_channel_TBME}(c) illustrate the accuracy and ITR of an RPT detector with $N_c = 8$ electrodes and $M = 9$ classes. This result is obtained using leave-p-out cross validation for $p = 12$, which indicates that only 20\% of the data is used for training. Note that IT CCA outperforms other methods 
if we increase the number of training trials for each class. Fig.~\ref{fig:ITCCA} shows the accuracy and ITR of the aforementioned methods for a multi-class multi-electrode setting using leave-one-out cross validation (i.e., we use the maximum number of training trials to construct the reference matrix. The reference matrix is defined later).
While IT CCA achieves the highest performance, it is important to note that it requires post-stimulus trials to construct the reference matrix. Hence, IT CCA boosts the performance of standard CCA at the expense of additional training with post-stimulus data per subject. By contrast, the RPT detector is unsupervised, in the sense that it only uses training to estimate the spatial correlation matrix $\mathbf{\Sigma_w}$ per subject from pre-stimulus data. Our results also indicate that the RPT detector can outperform PSDA and standard CCA methods for short data lengths, i.e., less than 1.5 seconds.
\begin{figure*}
\vspace{-.5cm}
\centering
\includegraphics[width=\textwidth,height=5cm]{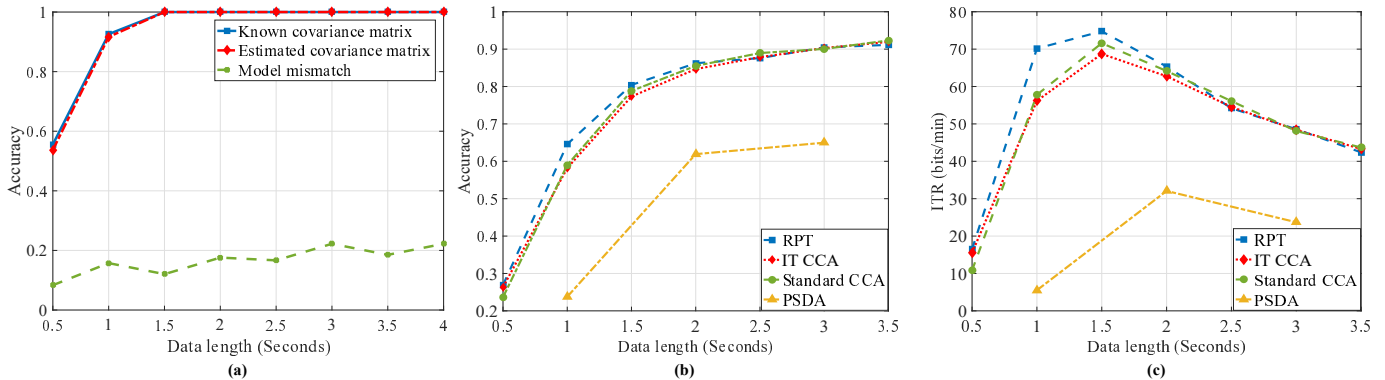}
\vspace{-.75cm}
\caption{Number of classes and electrodes are $M=9$ and $N_c =8$, (a) Classification accuracy $A$ versus the data length $T$, synthesized data is used to generate the plots, they compare $A$ when the RPT detector knows the true $\mathbf{\Sigma_w}$ (dashed blue), $\mathbf{\Sigma_w}$ is unknown and the RPT detector estimates $\mathbf{\Sigma_w}$ using pre-stimulus data (dotted red), the RPT detector falsely assumes there is no spatial correlation between recorded data of different electrodes (dashed green) (b) Classification accuracy $A$ versus the data length $T$, real data is used to generate these plots, they compare the RPT detector, IT CCA, standard CCA, and PSDA. For the PSDA method, the data recorded by the best electrode (the one which yields the highest accuracy) is selected. (c) ITR versus the data length $T$, real data is used to generate these plots.}
\label{fig:M_class_N_channel_TBME}
\vspace{-.25cm}
\end{figure*}
\begin{figure}
\centering
\includegraphics[width=\linewidth,height=4cm]{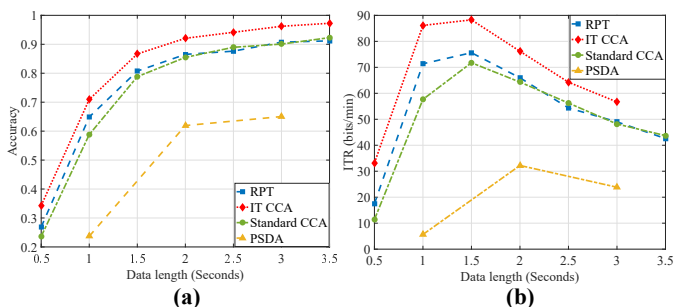}
\vspace{-.75cm}
\caption{(a) Classification accuracy and (b) ITR of the RPT detector, IT CCA, standard CCA and PSDA methods using leave-one-out cross validation.}
\label{fig:ITCCA}
\vspace{-0.25cm}
\end{figure}

\subsection{Difference between Standard CCA and RPT Detection Methods From Basis Representation Perspective}
In this section, we clarify a fundamental difference between standard CCA and the proposed RPT method and that is
the completeness of the basis used for signal representation. Standard CCA first constructs a reference matrix $\bQ_m \in \mathbb{R}^{2N_h\times L}$ from $N_h$ harmonics. Let $\bQ_m =
\begin{bmatrix}\bq_{m_1}&\ldots&\bq_{m_L} \end{bmatrix}$, where $\bq_{m_l}$ is a $2N_h \times 1$ column vector, defined as
\begin{equation}
    \begin{aligned}
    \bq_{m_l} \hspace{-3pt}=\hspace{-3pt} 
    \begin{bmatrix}
    \sin{\frac{\omega_ml}{f_s}}& \hspace{-6pt}\cos{\frac{\omega_ml}{f_s}} &\hspace{-8pt} \ldots &\hspace{-8pt} \sin{ \frac{N_h\omega_m l}{f_s}} &\hspace{-5pt} \cos{\frac{N_h\omega_m l}{f_s}}
    \end{bmatrix}^T
    \end{aligned}
\end{equation}
where $\omega_m = 2\pi f_m$ and $l = 1,2,\ldots,L$. For the measurement matrix $\bY$ and the reference matrix $\bQ_m$, standard CCA finds weight vectors $\bw_y$ ($N_c \times 1$) and $\bw_{q_m}$ ($2N_h \times 1$), which maximize the correlation between linear combinations of the signals recorded from the electrodes $\bz = \bY\bw_y$ and signals in the reference matrix $\bs_m = \bQ_m^T\bw_{q_m}$ (both are $L \times 1$ vectors) via solving the following optimization problem
\begin{equation}
\label{eq:CCA_opt}
    \begin{aligned}
    \rho_m = \max \rho(\bz,\bs_m) = \max_{\bw_y,\bw_{q_m}}\frac{\mathbb{E}[\bz\bs_m^T]}{\sqrt{\mathbb{E}[\bz\bz^T]\mathbb{E}[\bs_m\bs_m^T]}},
    \end{aligned}
\end{equation}
where $\rho_m$ is known as the maximum canonical correlation. Solving the optimization problem in (\ref{eq:CCA_opt}) $M$ times, once for each goal frequency $f_m$, and finding $\rho_m$'s, standard CCA obtains the decision using the following rule
\begin{equation}
    \begin{aligned}
    \hat{m} = \argmax_{m\in \mathcal{M}}\{\rho_0,\ldots,\rho_{M-1}\}.
    \end{aligned}
\end{equation}
However, such reference matrices do not form a finite complete basis for representing periodic signals. Hence, to be able to represent fully periodic signal, Standard CCA requires an infinite number of harmonics to construct $\bQ_m$ matrices. Therefore, the performance of Standard CCA depends on the number of harmonics used in the reference matrix. By contrast, in the RPT method we leverage a complete finite basis in the form of NPMs. To demonstrate this fact using synthesized data, we generate random periodic sequences with periods $T_0 = 32$  and $T_1 = 18$ that correspond to the two stimuli under hypotheses $H_0$  and $H_1$, respectively. We fix $\SNR =-12$ dB and generate 2000 observation vectors $\by$  under each hypothesis.
Fig.~\ref{fig:harmonics}(a) depicts the classification accuracy of the binary RPT detector and standard CCA versus different data lengths $T$ using the generated data. For the standard CCA method we consider three cases: in the first case we let the number of harmonics $N_h=1$ to generate the $\bQ_m$ reference matrices, in the second case $N_h=2$, and in the third case $N_h=3$. To produce Fig.~\ref{fig:harmonics}(b) we use the real data, and for standard CCA we consider two cases of $N_h=1$ and $N_h=2$. 
Both figures show that the accuracy of standard CCA method improves as we increase $N_h$ and enrich the reference matrices. This is in contrast to the RPT detection method which uses a finite complete basis for representing the periodic signals. 
\begin{figure}
\centering
\includegraphics[width=\linewidth,height=4cm]{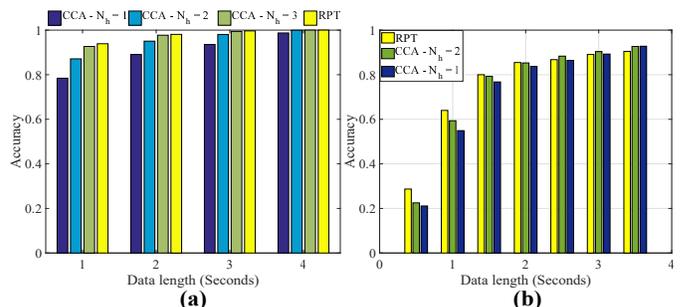}
\vspace{-0.75cm}
\caption{(a) Classification accuracy of standard CCA and the RPT detector on synthesized data, (b) real data.}
\label{fig:harmonics}
\vspace{-0.25cm}
\end{figure}

\subsection{Error Exponent- Discrimination Rate Tradeoff Based on Real and Synthesized Data}
In this section, we show the error exponent-discrimination rate tradeoff using both synthesized and real data.

\noindent$\bullet$\textit{Tradeoff based on synthesized data:} We generate random periodic sequences with period $T_m=T_0 + (m-1)\Delta T$, where $T_0=10$ and $\Delta T=1$ under the model in (\ref{eq:M_class_model_TBME}) for $M=11$ classes and $N_c = 8$ electrodes. We fix the data length $L = 50$ samples and $\SNR = -10$ dB and we generate 500 observations under each hypothesis.
We let the $(i,j)$-th entry of the spatial correlation matrix be $\mathbf{\Sigma_w}_{i,j}=\rho^{d_{ij}}$
for some $0 < \rho < 1$, where $d_{ij}$ denotes the distance between electrodes $i$ and $j$. Fig. \ref{fig:Mul_diversity_synz_TBME}(a) and \ref{fig:Mul_diversity_synz_TBME}(b) plot the error exponent $-\log P_e/(L\cdot\SNR)$ versus the discrimination rate $\log_2M$ using the generated synthesized data for $N_c = 4$ and $N_c = 8$, respectively. These figures show that the RPT detector achieves a better tradeoff compared to standard CCA. Fig. \ref{fig:Mul_diversity_synz_TBME}(c) plots the same using the synthesized data for the RPT detector for $N_c=1,2,4,8$. This figure shows that a larger $N_c$ yields a better tradeoff between the error exponent and the discrimination rate.

\noindent$\bullet$\textit{Tradeoff based on real data:} Fig.~\ref{fig:Mul_diversity_real_TBME}(a) illustrates the tradeoff using real data for $L= 256$ ($T= 1$ second) and $N_c=8$, indicating that the RPT detector exhibits a better tradeoff compared to standard CCA and IT CCA methods. Fig.~\ref{fig:Mul_diversity_real_TBME}(b) plots the same using real date for the RPT detector for $N_c=1,2,4,8$, verifying that a larger $N_c$ provides a better tardeoff between the error exponent and the discrimination rate. In these experiments $\SNR$ is estimated from pre-stimulus and post-stimulus data.
\begin{figure*}
\centering
\includegraphics[width=\textwidth,height=5.5cm]{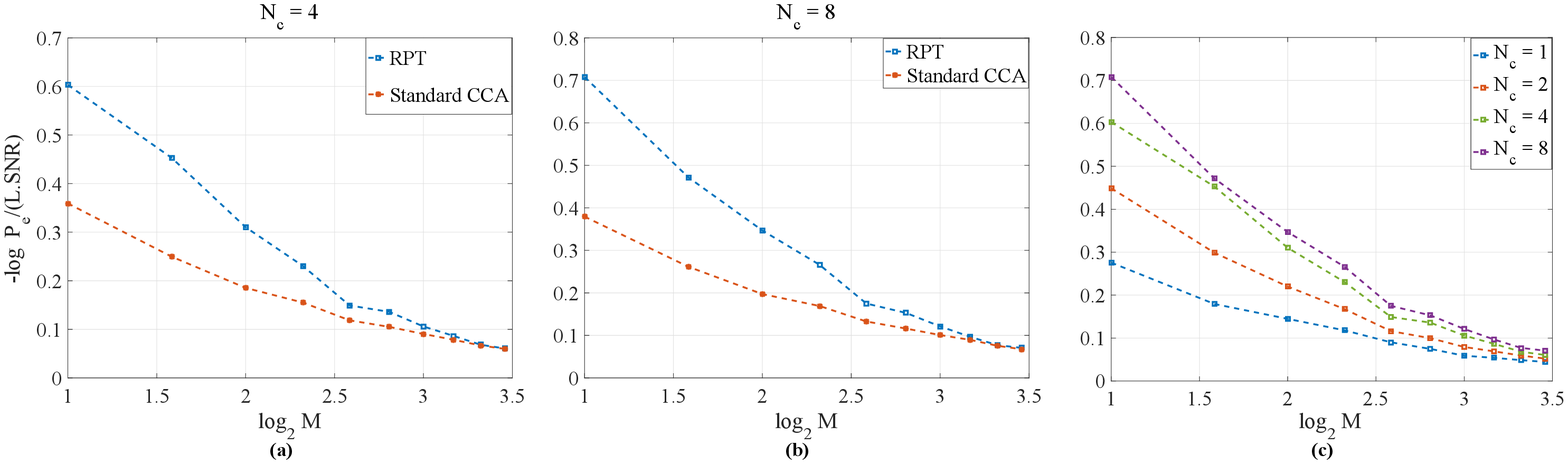}
\vspace{-.75cm}
\caption{(a) and (b) Error exponent-discrimination rate tradeoff using synthesized data for $N_c = 4$ and $N_c = 8$ respectively, with $M$ changing from 2 to 11 with fixed $L=50$ and $\SNR = -10$ dB. (c)  Error exponent-discrimination rate tradeoff for the RPT detector using synthesized data for different $N_c$.}
\label{fig:Mul_diversity_synz_TBME}
\vspace{-.5cm}
\end{figure*}

\begin{figure}
\centering
 \includegraphics[width=\linewidth,height=4cm]{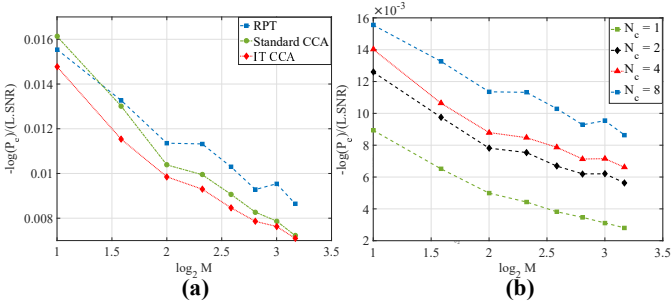}
\vspace{-.7cm}
\caption{(a) Error exponent-discrimination rate tradeoff using real data with $M$ changing from 2 to 9 with fixed $L=256$ and $N_c = 8$. (b)   Error exponent-discrimination rate tradeoff of the RPT detector for different $N_c$.}
\label{fig:Mul_diversity_real_TBME}
\vspace{-.5cm}
\end{figure}

\section{Discussion}  \label{sec:disc}
We started with an in depth analysis of a composite binary hypothesis testing framework to detect periodic SSVEP responses. The RPT detector distinguishes $M$ classes (corresponding to $M$ hypotheses $H_m$) using the Ramanujan subspace, i.e., the RPT detector chooses the class corresponding to the subspace that yields the largest projection. We remark that NPMs are designed for integer periodicity estimation. Hence, we have to round the periods corresponding to the stimuli frequencies to the nearest integer values. This is the main drawback of employing the RPT detector for SSVEP detection since we may not have exact integer periods.
Given the model in section \ref{sec:M_ary model}, the Ramanujan matrix $\bK_{S_m}$  corresponding to $T_m$ contains the submatrices $\bR_m$ and $\bR_{d|m}$. Therefore, the submatrices within, inherently span the subspace of $T_m$ and its divisors. If $T_m$ is an even integer, its divisors correspond to even harmonics of the goal frequency $f_m$ (i.e., $nf_m$ where $n$ is even), which can enhance the performance of the RPT detector. 

We also compared the performance of the RPT detector to that of a fictitious detector that knows $\bx_0$ and $\bx_1$, i.e., the signal representations in the RPT dictionary. This detector yields an upper bound on the performance for the composite test. As shown in Fig.~\ref{fig:GAP_TBME}, the RPT detector based on the GLRT is asymptotically optimal as it closes the gap to the perfect measurement bound as $L$ increases.

The proposed RPT detection method is unsupervised since it does not use any information from post-stimulus data and uses only pre-stimulus data to estimate the covariance matrix capturing the spatial correlation between the recorded data of different electrodes. 
The performance of supervised methods is heavily dependent on the number of training trials. For instance, for the IT CCA method it is crucial to have enough training data to construct the reference matrices and it has been shown that reducing the number of training trials can deteriorate the performance of the method.

In contrast to the reference matrix in standard CCA which does not provide a complete basis for periodic signals, the RPT detector leverages a complete dictionary spanning the subspaces of all periodic signals.




\section{Conclusion}  \label{sec:conc}
We proposed and analyzed a new approach to SSVEP detection using linear representations in RPT dictionaries known to be robust to noise and latency. The RPT detector outperforms the state-of-the-art methods in the short data length regime crucial for real-time BCI. Further, it does not depend on post-stimulus data, which reduces the overhead associated with data collection in supervised methods.   
Furthermore, we introduced a new tradeoff between the error exponent and the discrimination rate, which can serve as a basis for comparing SSVEP detection schemes across an entire range of operation regimes where efficiency is traded for rate and vice-versa. 

\appendices
\section{Proof of Theorem \ref{thr:dist_theorem}} 
\label{App:proof_thr_dist}
From \cite[Lemma 1.1]{anderson1980cochran}, the RVs $\by^T\bB^{\perp}
\by$ and $\by^T \bA^{\perp} \by$ in (\ref{eq:dec_orthogonal_TBME}) have non-central Chi-squared distributions with $r_B^{\perp}$ and $r_A^{\perp}$ degrees of freedom and non-centrality parameters $\lambda_{m,B}^{2,\perp}$ and $\lambda_{m,A}^{2,\perp}$, respectively, under hypothesis $H_m$, where $r_B^{\perp}=\tr(\bB^{\perp})$, $r_A^{\perp}=\tr(\bA^{\perp})$, $\lambda_{m,B}^{2,\perp} = \boldsymbol{\mu}_m^T \bB^{\perp} \boldsymbol{\mu}_m$, $\lambda_{m,A}^{2,\perp} = \boldsymbol{\mu}_m^T \bA^{\perp} \boldsymbol{\mu}_m$ and $\boldsymbol{\mu}_m$ is the mean of the observation $\by$ under $H_m$. Moreover, we have shown in Lemma \ref{lem:A_B_independent} that $\bA^{\perp}\bB^{\perp} = \mathbf{0}$. Hence, it follows from \cite{cochran1934distribution} that $\by^T\bB^{\perp}  
\by$ and $\by^T \bA^{\perp} \by$ are independent. From the definitions above and the orthogonality of the submatrices associated with different divisors, we can readily show that $\lambda_{1,A}^{2,\perp} = \lambda_{0,B}^{2,\perp}=0$. Therefore, the test statistic $\ell(\by)$ in (\ref{eq:dec_orthogonal_TBME}) is the difference between a non-central and a central Chi-squared RV. Per \cite[Theorem 1.1]{lancaster1969chi}, a non-central Chi-squared RV with $n$ degrees of freedom can be represented as a sum of a non-central Chi-squared RV with one degree of freedom with the same non-centrality parameter and a central Chi-squared RV with $n-1$ degrees of freedom. Moreover, we can write a central Chi-squared RV with $(m+n)$ degrees of freedom as the sum of two independent Chi-squared RVs with $m$ and $n$ degrees of freedom. Accordingly, under $H_1$,
\begin{equation}
\label{eq:ly_as_linearcomb}
\begin{aligned}
\ell(\by)\hspace{-1pt}= \ell_1(\by) + \ell_2(\by) - \ell_3(\by) - \ell_4(\by)
\end{aligned}
\end{equation}
where $\ell_1(\by)\sim \chi^2(1,\lambda_{1,B}^{2,\perp})$, $\ell_2(\by)\sim \chi^2(r_B^{\perp}-1,0)$,$\ell_3(\by)\sim \chi^2(1,0)$ and $\ell_4(\by)\sim \chi^2(r_A^{\perp}-1,0)$.
Using \cite[Theorem 3.3]{press1966linear} which characterizes the distribution of the difference of linear combinations of non-central Chi-squared RVs, we can express the distribution of $\ell(\by)$ as 
\begin{equation}
\begin{aligned}
h(t) = \sum_{i=0}^\infty\sum_{j=0}^\infty q_i l_j p_{r^\perp_B+2i,r_A^\perp+2j}(t),
\end{aligned}
\end{equation}
where $p_{a,b}(.)$ is defined in (\ref{eq:dist_final_TBME}), and following from \cite[Theorem 2.1A]{press1966linear} and (\ref{eq:ly_as_linearcomb}), we can readily obtain the coefficients $q_i$'s and $l_j$'s as
\begin{equation}
\begin{aligned}
q_0 = \exp\left(\frac{-\lambda_{1,B}^{2,\perp}}{2}\right), \qquad l_0 = 1,
\end{aligned}
\end{equation}
\begin{equation}
\begin{aligned}
q_i = \frac{\exp\left({ \frac{-\lambda_{1,B}^{2,\perp}}{2}}\right)\left(\frac{\lambda_{1,B}^{2,\perp}}{2}\right)^{i}}{i!} \quad \text{for } i\neq 0\:,
\end{aligned}
\end{equation}
and $l_j = 0$ for $j > 0$. Similarly, we can derive the corresponding coefficients under $H_0$. Based on the aforementioned definitions for the non-centrality parameters, we have
\begin{equation}
\begin{aligned}
\lambda_{0,A}^{2,\perp}=& \bx_{S_0}^{T}\bK_{S_0}^T\bA^{\perp} \bK_{S_0}\bx_{S_0}\\
\lambda_{1,B}^{2,\perp}=&\bx_{S_1}^{T}\bK_{S_1}^T\bB^{\perp} \bK_{S_1}\bx_{S_1}
\end{aligned}
\end{equation}
which subsequently leads to the equations in (\ref{eq:char_ortg_TBME}).
\section{Proof of (\ref{eq:cross_cov_TBME})} \label{App:proof_cov}
Here, we analyze the covariance of the two quadratic terms in the expression of the test statistic $\ell(\by)$ in (\ref{eq:Decision_rule_general_binary_TBME}) under $H_0$.

\begin{equation}
\label{eq:cov_nonorthogonal}
\begin{aligned}
\cov(&\by^T \bB \by ,\by^T \bA \by|H_0) \\
= & \mathbb{E}_0[\by^T \bB \by\by^T \bA \by]-\mathbb{E}_0[\by^T \bB \by]\mathbb{E}_0[\by^T \bA \by]
\end{aligned}
\end{equation}
where $\mathbb{E}_m$ denotes the expectation under $H_m$. 
Focusing on the first term in (\ref{eq:cov_nonorthogonal}) we expand it as below:
\begin{equation}
\begin{aligned}
&\mathbb{E}_0[\by^T \bB \by\by^T \bA \by]= \\
&\bx_{S_0}^T\bK_{S_0}^T\bB\bK_{S_0}\bx_{S_0}\bx_{S_0}^T\bK_{S_0}^T\bK_{S_0}\bx_{S_0}+\mathbb{E}[\bw^T\bB\bw\bw^T\bA\bw]\\& + \bx_{S_0}^T\bK_{S_0}^T\bB \mathbb{E}[\bw\bx_{S_0}^T\bK_{S_0}^T\bw] + \bx_{S_0}^T\bK_{S_0}^T\bB \mathbb{E}[\bw\bw^T]\bK_{S_0}\bx_{S_0}\\& + \bx_{S_0}^T\bK_{S_0}^T\bB \mathbb{E}[\bw\bw^T\bA\bw] + \mathbb{E}[\bw^T\bB\bK_{S_0}\bx_{S_0}\bx_{S_0}^T\bK_{S_0}^T\bw]\\
& + \mathbb{E}[\bw^T\bB\bK_{S_0}\bx_{S_0}\bw^T]\bK_{S_0}\bx_{S_0} + \mathbb{E}[\bw^T\bB\bK_{S_0}\bx_{S_0}\bw^T\bA\bw]\\&+\mathbb{E}[\bw^T\bB\bw]\bx_{S_0}^T\bK_{S_0}^T\bK_{S_0}\bx_{S_0}
+\mathbb{E}[\bw^T\bB\bw\bx_{S_0}^T\bK_{S_0}^T\bw]\\&+\mathbb{E}[\bw^T\bB\bw\bw^T]\bK_{S_0}\bx_{S_0}+\bx_{S_0}^T\bK_{S_0}^T\bB\bK_{S_0}\bx_{S_0} \mathbb{E}[\bw^T\bA\bw]
\end{aligned}
\end{equation}
Using the definitions in (\ref{eq:chi_char_TBME}), it simplifies to
\begin{equation}
\label{eq:cov_sec_m}
\begin{aligned}
\mathbb{E}_0[\by^T \bB \by\by^T \bA \by]&=\lambda_{0,B}^2\lambda_{0,A}^2+r_Br_A +2 \sum_{i=1}^L\sum_{j=1}^L c_{ij} \\& + 4\lambda_{0,B}^2 +r_B\lambda_{0,A}^2 +r_A\lambda_{0,B}^2
\end{aligned}
\end{equation}
where $c_{ij}$ are the entries of the matrix $\bC= \bA \odot \bB$, where
$\odot$ denotes the element-wise product. Similarly, we can show that
\begin{equation}
\label{eq:eq:cov_first_m}
\begin{aligned}
&\mathbb{E}_0[\by^T \bB\by]\mathbb{E}_0[\by^T \bA\by]  = (\lambda_{0,B}^2 + r_B)(\lambda_{0,A}^2+r_A)
\end{aligned}
\end{equation}
Substituting (\ref{eq:cov_sec_m}) and (\ref{eq:eq:cov_first_m}) into (\ref{eq:cov_nonorthogonal}) we find
\begin{equation}
\begin{aligned}
\cov(\by^T \bB \by,\by^T \bA \by|H_0) = 2 \sum_{i=1}^L\sum_{j=1}^L c_{ij}+ 4\lambda_{0,B}^2 
\end{aligned}
\end{equation}
The term $\cov(\by^T \bB \by,\by^T \bA \by|H_1)$ can be derived similarly.


\section{Proof of Lemma \ref{lem:Gap}} \label{App:lemma_gap}
From (\ref{eq:SNR}) We have
\begin{equation}
\begin{aligned}
L\cdot\SNR_m =& \bx_{S_m}^T\bK_{S_m}^T\bK_{S_m}\bx_{S_m}.\\
\end{aligned}
\end{equation}
Based on the definitions of  $\lambda_{0,A}^2$ and $\lambda_{1,B}^2$ in (\ref{eq:chi_char_TBME}), we have $L\cdot\SNR_0 = \lambda_{0,A}^2$ and $L\cdot\SNR_1 = \lambda_{1,B}^2$. Borrowing the Q-function approximation in \cite{borjesson1979simple} given below
\begin{equation}\label{eq:Q_apprx}
\begin{aligned}
Q(x) \approx \frac{e^{-x^2/2}}{\sqrt{2\pi}\sqrt{1+x^2}} \quad \text{for} \quad x>0,
\end{aligned}
\end{equation}


we can approximate (\ref{eq:P_D_NP_binary_TBME}) and (\ref{eq:P_D_PMB_TMBE}) for large values of $L\cdot\SNR_m$ as the following
\begin{equation}
\begin{aligned}
P_{D} \approx 1 - \frac{e^{-\frac{L(\SNR_0+\SNR_1)^2}{8\SNR_1}}}{\sqrt{2\pi}\sqrt{\frac{L(\SNR_0+\SNR_1)^2}{4\SNR_1}}}
\end{aligned}
\end{equation}
and
\begin{equation}
\begin{aligned}
P_{D_{\text{PMB}}} \approx 1- \frac{e^{-\frac{L(\SNR_0+\SNR_1)}{2}}}{\sqrt{2\pi}\sqrt{L(\SNR_0+\SNR_1)}}.
\end{aligned}
\end{equation}
Without loss of generality we assume $\SNR = \SNR_0 = \SNR_1$, and thus we have $L\cdot\SNR = L\cdot\SNR_0 = L\cdot \SNR_1$. Using this assumption we can write: 
\begin{equation}
\begin{aligned}
P_{D} \approx 1 - \frac{e^{-\frac{L\cdot\SNR}{2}}}{\sqrt{2\pi}\sqrt{L\cdot\SNR}}
\end{aligned}
\end{equation}
and
\begin{equation}
\begin{aligned}
P_{D_{\text{PMB}}} \approx 1- \frac{e^{-L.\SNR}}{\sqrt{2\pi}\sqrt{2L.\SNR}}.
\end{aligned}
\end{equation}
Therefore, the gap is approximated by
\begin{equation}
\begin{aligned}
\gap(L,\SNR)  \approx&\frac{e^{-\frac{L\cdot\SNR}{2}}}{\sqrt{2\pi}\sqrt{L\cdot\SNR}} - \frac{e^{-L\cdot\SNR}}{\sqrt{2\pi}\sqrt{2L\cdot\SNR}}\\
=&\frac{e^{-\frac{L\cdot\SNR}{2}}\left(\sqrt{2}-e^{-\frac{L\cdot\SNR}{2}}\right)}{2\sqrt{\pi}\sqrt{L\cdot\SNR}}
\end{aligned}
\end{equation}
with the asymptotic order in (\ref{Gap_order_TBME}).

\section{} \label{App:est}
Let $\mathcal{D}$ denote the first derivative of the expression in (\ref{eq:ML_X_S_TBME}) with respect to $\bX_{S_m}$. The estimate $\hat{\bX}_{S_m}$ can be obtained by setting $\mathcal{D}$ equal to zero. It is easy to verify the following \cite{petersen2008matrix},
\begin{equation}
\begin{aligned}
\mathcal{D}=\frac{\partial}{\partial \bX_{S_m}}& \hspace{-.5mm}\tr\left(-2\bY\mathbf{\Sigma_w}^{-1}\bX_{S_m}^{T}\bK_{S_m}^{T}\right) + \\
\frac{\partial}{\partial \bX_{S_m}}& \hspace{-.5mm}\tr\hspace{-.5mm}\left(
\bK_{S_m}\bX_{S_m}\mathbf{\Sigma_w}^{-1}\bX_{S_m}^T\bK_{S_m}^T\right) \hspace{-.5mm}\\ 
=& -2\bK_{S_m}^{T}\bY\mathbf{\Sigma_w}^{-1}+2\bK_{S_m}^T\bK_{S_m}{\bX_{S_m}}\mathbf{\Sigma_w}^{-1}\nonumber
\end{aligned}
\label{eq:trace_der}
\end{equation}
Letting $\mathcal{D}$ equal to zero and solving for $\bX_{S_m}$ we find
\begin{equation} \nonumber
\begin{aligned}
\hat{\bX}_{S_m} = \left(\bK_{S_m}^{T}\bK_{S_m}\right)^{-1}\bK_{S_m}^{T}\bY.
\end{aligned}
\end{equation}

\section*{Acknowledgments}
This work was supported in part by NSF Grants CCF-1525990, CCF-1552497 and CCF-1341966.
\bibliographystyle{IEEEtran}
\bibliography{bibfile}

%
\end{document}